\documentclass[preprint,authoryear]{elsarticle}
\usepackage[T1]{fontenc}
\usepackage[utf8]{inputenc}

\biboptions{sort,longnamesfirst}
\renewcommand\cite[1]{\citep{#1}}

\usepackage{amsopn}
\newtheorem{theorem}{Theorem}
\newtheorem{lemma}[theorem]{Lemma}
\newtheorem{proposition}[theorem]{Proposition}
\newtheorem{corollary}[theorem]{Corollary}
\newdefinition{definition}[theorem]{Definition}
\newdefinition{remark}[theorem]{Remark}
\newdefinition{example}[theorem]{Example}
\newproof{proof}{Proof}

\newcommand{\R}{\mathbb{R}}
\newcommand{\Q}{\mathbb{Q}}
\newcommand{\N}{\mathbb{N}}
\newcommand{\Z}{\mathbb{Z}}
\newcommand{\K}{\mathbb{K}}
\newcommand{\MAT}[3]{M_{#1\ifthenelse{\equal{#2}{}}{}{,#2}}\ifthenelse{\equal{#3}{}}{}{\left(#3\right)}}
\newcommand{\Rgen}{\R_{G}}
\newcommand{\Rpoly}{\R_{P}}
\newcommand{\Rp}{\R_{\geqslant0}}
\newcommand{\intinterv}[2]{\llbracket#1,#2\rrbracket}
\newcommand{\powerset}[1]{\mathcal{P}(#1)}
\newcommand{\dom}{\operatorname{dom}}
\newcommand{\card}[1]{{\##1}}

\newcommand{\inorm}[2]{\left\lVert{#1}\right\rVert_{#2}}
\newcommand{\twonorm}[1]{\inorm{#1}{2}}
\newcommand{\infnorm}[1]{\inorm{#1}{}}

 \newcommand{\openball}[2]{B_{#2}(#1)}

\newcommand{\sigmap}[1]{{\Sigma{#1}}}
\newcommand{\degp}[1]{{\operatorname{deg}(#1)}}
\newcommand{\poly}{\operatorname{poly}}

\newcommand{\sgn}[1]{\operatorname{sgn}(#1)}
\newcommand{\intp}{\operatorname{int}}
\newcommand{\fiter}[2]{#1^{[#2]}}
\newcommand{\frestrict}[2]{#1\restriction_{#2}}
\newcommand{\idfun}{\operatorname{id}}

\newcommand{\round}[1]{\left\lfloor#1\right\rceil}

\newcommand{\myclass}[1]{\operatorname{#1}}
 \newcommand{\gval}[2][]{\ensuremath{\myclass{GVAL}_{#1}\ifthenelse{\equal{#2}{}}{}{[#2]}}}

\newcommand{\lambdafun}[2]{(#1)\mapsto{#2}}
\newcommand{\lambdafunex}[2]{#1\mapsto{#2}}

\newcommand{\jacobian}[1]{J_{#1}}

\newcommand{\scalarprod}[2]{{#1}\cdot{#2}}

\newcommand{\mtt}[1]{\mathtt{#1}}
\newcommand{\ovl}[1]{\overline{#1}}

\newcommand{\myop}[1]{\operatorname{#1}}
\newcommand{\abs}{\myop{abs}}
\newcommand{\sg}{\myop{sg}}
\newcommand{\ip}[1]{\myop{ip}_{#1}}
\newcommand{\rnd}{\myop{rnd}}

\newcommand{\lxh}{\myop{lxh}}
\newcommand{\hxl}{\myop{hxl}}

\newcommand{\norm}{\myop{norm}}
\newcommand{\mx}{\myop{mx}}
\newcommand{\mn}{\myop{mn}}

\newcommand{\clamp}{\myop{clamp}}

\newcommand{\fnsg}[3]{tanh((#1)*(#2)*(#3))}
\newcommand{\fnipone}[3]{(1+\fnsg{(#1)-1}{#2}{#3})/2.}

\newcommand{\fnlxh}[5]{\fnipone{(#3)-(#1+#2)/2.+1}{#4+log(1+(#5)*(#5))}{2./(#2-#1)}*(#5)}
\newcommand{\fnhxl}[5]{\fnipone{(#1+#2)/2.-(#3)+1}{#4+log(1+(#5)*(#5))}{2./(#2-#1)}*(#5)}

\usepackage[english]{babel}
\usepackage{amsmath,amssymb,mathrsfs,amsfonts}
\usepackage{mathtools}
\usepackage{mathdots}
\usepackage{color}
\usepackage{xifthen}
\usepackage{stmaryrd}
\usepackage{tikz}
\usepgflibrary{arrows}
\usetikzlibrary{fit}
\usetikzlibrary{patterns}
\usetikzlibrary{shapes.geometric}
\usepackage{environ}
\usepackage{longtable}
\usepackage{algpseudocode}
\usepackage{dsfont}
\usepackage{thmtools}
\usepackage[
  pdftex,
  bookmarks=true,
  bookmarksnumbered=true,
  hypertexnames=false,
  breaklinks=true,
  linkbordercolor={0 0 1}
  ]{hyperref}
\usepackage{todonotes}
\usepackage{fancyhdr}
\usepackage{booktabs}
\usepackage{epigraph}
\usepackage{enumitem}

\allowdisplaybreaks

\newcommand{\lemref}[1]{Lemma~\ref{#1}}
\newcommand{\thref}[1]{Theorem~\ref{#1}}
\newcommand{\amaurycorref}[1]{Corollary~\ref{#1}}
\newcommand{\figref}[1]{Figure~\ref{#1}}

\newcommand{\propref}[1]{Proposition~\ref{#1}}

\newcommand{\remref}[1]{Remark~\ref{#1}}

\title{On the Functions Generated by the General Purpose Analog Computer}

\begin{document}

\author[lix]{Olivier Bournez\corref{cor}\fnref{dga}}
\ead{bournez@lix.polytechnique.fr}
\author[fct,sqig]{Daniel Gra\c{c}a\fnref{feder}}
\ead{dgraca@ualg.pt}
\author[lix,mpi]{Amaury Pouly\fnref{dga}}
\ead{pamaury@lix.polytechnique.fr}

\address[lix]{École Polytechnique, LIX, 91128 Palaiseau Cedex, France}
\address[fct]{FCT, Universidade do Algarve, C. Gambelas, 8005-139 Faro, Portugal}
\address[sqig]{Instituto de Telecomunica\c{c}\~{o}es, Lisbon, Portugal}
\address[mpi]{MPI-SWS, E1 5, Campus, 66123 Saarbr\"{u}cken, Germany }
\cortext[cor]{Corresponding author}
\fntext[feder]{Daniel Gra\c{c}a was partially supported by
  \emph{Funda\c{c}\~{a}o para a Ci\^{e}ncia e a Tecnologia} and EU FEDER
  POCTI/POCI via SQIG - Instituto de Telecomunica\c{c}\~{o}es through
  the FCT project UID/EEA/50008/2013. This project has received funding from the European Union’s Horizon 2020 research and innovation
  programme under the Marie Skłodowska-Curie grant agreement No 731143.}
\fntext[dga]{Olivier Bournez and Amaury Pouly were partially supported by
  \emph{DGA Project CALCULS}}

\begin{abstract}
  We consider the General Purpose Analog Computer (GPAC), introduced by
  Claude Shannon in 1941 as a mathematical
  model of Differential Analysers, that is to say as a model of
  continuous-time analog (mechanical, and later one electronic)
  machines of that time. 

  The GPAC generates as output univariate functions  (i.e. functions $f: \R \to \R$).
  In this paper we extend this model by: (i) allowing multivariate functions (i.e. functions
  $f: \R^n \to \R^m$); (ii) introducing a notion of amount of
  resources (space) needed to generate a function, which allows the
  stratification of GPAC generable functions into proper subclasses.
  We also prove that a wide class of (continuous and
  discontinuous) functions can be uniformly approximated over their
  full domain.

  We prove a few stability properties of this
  model, mostly stability by arithmetic operations, composition and
  ODE solving, taking into account the amount of resources needed to perform each operation.
  
  We establish that generable functions are always
  analytic but that they can nonetheless (uniformly) approximate a wide range
of nonanalytic functions. Our model and results extend some of the results from
\cite{Sha41} to the multidimensional case, allow one to define classes of functions generated by GPACs which take into account bounded resources, and also strengthen the
approximation result from \cite{Sha41} over a compact domain to a
uniform approximation result over unbounded domains.

\end{abstract}


\maketitle

\section{Introduction}

In 1941, Claude Shannon introduced in \cite{Sha41} the GPAC
model as a model for the Differential Analyzer \cite{Bus31}, which are
mechanical (and later on electronics) continuous time analog machines,
on which he worked as an operator. The model was later refined in
\cite{Pou74}, \cite{GC03}.  Originally it was presented as a model
based on circuits. Basically, a GPAC is any circuit that can be build
from the 4 basic units of Figure \ref{fig:gpac_circuit}, that is to
say from basic units realizing constants, additions, multiplications
and integrations, all of them working over analog real quantities
(that were corresponding to angles in the mechanical Differential
Analysers, and later on to voltage in the electronic versions).

\begin{figure}
\begin{center}
 \setlength{\unitlength}{1200sp}%
\begin{tikzpicture}
 \begin{scope}[shift={(0,0)},rotate=0]
  \draw (0,0) -- (0.7,0) -- (0.7,0.7) -- (0,0.7) -- (0,0);
  \node at (.35,.35) {$k$};
  \draw (.7,.35) -- (1,.35);
  \node[anchor=west] at (1,.35) {$k$};
  \node at (.35, -.3) {A constant unit};
 \end{scope}
 \begin{scope}[shift={(4,0)},rotate=0]
  \draw (0,0) -- (0.7,0) -- (0.7,0.7) -- (0,0.7) -- (0,0);
  \node at (.35,.35) {$+$};
  \draw (.7,.35) -- (1,.35); \draw (-.3,.175) -- (0,.175); \draw (-.3,.525) -- (0,.525);
  \node[anchor=west] at (1,.35) {$u+v$};
  \node at (.35, -.3) {An adder unit};
  \node[anchor=east] at (-.3,.525) {$u$};
  \node[anchor=east] at (-.3,.175) {$v$};
 \end{scope}
 \begin{scope}[shift={(0,-2)},rotate=0]
  \draw (0,0) -- (0.7,0) -- (0.7,0.7) -- (0,0.7) -- (0,0);
  \node at (.35,.35) {$\times$};
  \draw (.7,.35) -- (1,.35); \draw (-.3,.175) -- (0,.175); \draw (-.3,.525) -- (0,.525);
  \node[anchor=west] at (1,.35) {$uv$};
  \node at (.35, -.3) {An multiplier unit};
  \node[anchor=east] at (-.3,.525) {$u$};
  \node[anchor=east] at (-.3,.175) {$v$};
 \end{scope}
 \begin{scope}[shift={(4,-2)},rotate=0]
  \draw (0,0) -- (0.7,0) -- (0.7,0.7) -- (0,0.7) -- (0,0);
  \node at (.35,.35) {$\int$};
  \draw (.7,.35) -- (1,.35); \draw (-.3,.175) -- (0,.175); \draw (-.3,.525) -- (0,.525);
  \node[anchor=west] at (1,.35) {$w=\int u\thinspace dv$};
  \node at (.35, -.3) {An integrator unit};
  \node[anchor=east] at (-.3,.525) {$u$};
  \node[anchor=east] at (-.3,.175) {$v$};
 \end{scope}
\end{tikzpicture}
\end{center}
\caption{Circuit presentation of the GPAC: a circuit built from basic
  units. Presentation of the 4 types of units: constant, adder,
  multiplier, and integrator.}
\label{fig:gpac_circuit}
\end{figure}
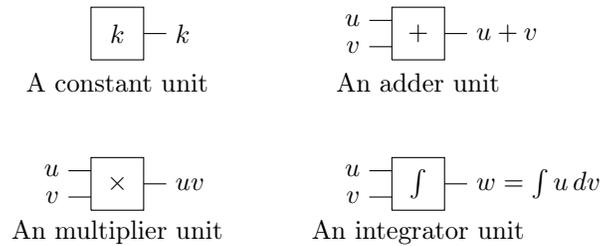%

Figures~\ref{fig:gpac_example_sin} illustrates for example how the
sine function can generated using two integrators, with suitable
initial state, as being the
solution of ordinary differential equation $$\left\{\begin{array}{@{}c@{}l}
y'(t)&=z(t)\\
z'(t)&=-y(t)\\
\end{array}
\right.$$
with suitable initial conditions.

\begin{figure}
\begin{center}
\begin{tikzpicture}
 \begin{scope}[shift={(-4.2,0)},rotate=0]
  \draw (0,0) -- (0.7,0) -- (0.7,0.7) -- (0,0.7) -- (0,0);
  \node at (.35,.35) {$-1$};
 \end{scope}
 \draw (-3.5,.35) -- (-3.15,.35) -- (-3.15,.175) -- (-2.8,.175);
 \begin{scope}[shift={(-2.8,0)},rotate=0]
  \draw (0,0) -- (0.7,0) -- (0.7,0.7) -- (0,0.7) -- (0,0);
  \node at (.35,.35) {$\times$};
 \end{scope}
 \draw (-2.1,.35) -- (-1.75,.35) -- (-1.75,.525) -- (-1.4,.525);
 \begin{scope}[shift={(-1.4,0)},rotate=0]
  \draw (0,0) -- (0.7,0) -- (0.7,0.7) -- (0,0.7) -- (0,0);
  \node at (.35,.35) {$\int$};
 \end{scope}
 \draw (-.7,.35) -- (-.35,.35) -- (-.35,.525) -- (0,.525);
 \begin{scope}[shift={(0,0)},rotate=0]
  \draw (0,0) -- (0.7,0) -- (0.7,0.7) -- (0,0.7) -- (0,0);
  \node at (.35,.35) {$\int$};
 \end{scope}
 \draw (.7,.35) -- (1.4,.35);
 \node[anchor=west] at (1.4,.35) {$\sin(t)$};
 \node[anchor=north] at (-1, -0.5) {$\left\lbrace
\begin{array}{@{}c@{}l}
y'(t)&=z(t)\\
z'(t)&=-y(t)\\
y(0)&=0\\
z(0)&=1
\end{array}
\right.\Rightarrow\left\lbrace\begin{array}{@{}c@{}l}
y(t)&=\sin(t)\\
z(t)&=\cos(t)
\end{array}\right.$};
 \draw (1,.35) -- (1,1) -- (-3.15,1) -- (-3.15,.525) -- (-2.8,.525);
 \fill (1,.35) circle[radius=.07];
 \draw (-4.9,.35) -- (-4.55,.35) -- (-4.55,-.3) -- (-.3,-.3) -- (-.3,.175) -- (0,.175);
 \draw (-1.75,-.3) -- (-1.75,.175) -- (-1.4,.175);
 \fill (-1.75,-.3) circle[radius=.07];
 \node[anchor=east] at (-4.9,.35) {$t$};
\end{tikzpicture}
\end{center}
\caption{Example of GPAC circuit: computing sine and cosine with two variables}
\label{fig:gpac_example_sin}
\end{figure}
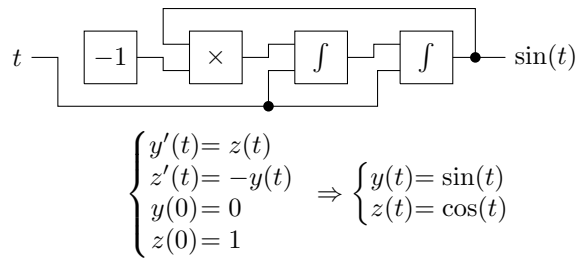

The original GPAC model introduced by Shannon has the feature that it works
in \emph{real time}: for example if the input $t$ is updated in the GPAC circuit of
Figure \ref{fig:gpac_example_sin}, then the output is immediately updated for the corresponding value of $t$.

Shannon himself realized that functions computed by a GPAC
are nothing more than solutions of a special class of polynomial
differential equations. In particular it can be shown that a function
$f: \R \to \R$ is computed by Shannon's model \cite{Sha41},
\cite{GC03} if and only if it is a (component of the) solution of a
polynomial initial value problem of the form 

\begin{equation}\label{eq:ode}
\left\{\begin{array}{@{}r@{}l}y'(t)&=p(y(t))\\y(t_0)&=y_0\end{array}\right.\qquad t\in\R
\end{equation}
where $p$ is a vector of polynomials and $y(t)$ is vector. In other
words, $f(t)=y_1(t)$, and $y_i'(t)=p_i(y(t))$
where $p_i$ is a multivariate polynomial.

Basically, the idea is just to introduce a variable for each output of
a basic unit, and write the corresponding ordinary differential
equation (ODE), and observe that it can be written as an ODE
with a polynomial right hand side. 

\begin{remark}
Technically speaking, the initial model studied by Claude Shannon in
\cite{Sha41} suffers from problems of being sometimes not fully
formally defined and some key proofs in that paper contained imprecisions. This has been observed and
refined later in several papers in particular in \cite{Pou74} in order
to get to a model where the result of Shannon stating the equivalence
of computable functions with differentially algebraic functions precisely hold. However,
the paper \cite{Pou74} had the problem that the GPAC model it presented had no direct
connection to circuits built using the units of Figure \ref{fig:gpac_circuit}, and therefore seemed
to lack the physical resemblance to Differential analysers (see \cite{GC03} for a discussion).
In the paper \cite{GC03} these problems are solved by formally defining rules to get allowable GPAC
circuits (removing bizarre possibilities that could happen in Shannon's original model
like linking the output of an adder unit to one of its inputs) which ensure that each GPAC
circuit has one or more outputs, which exist and are unique.
Moreover, the GPAC defined as in \cite{GC03} seems to capture all the functions computed
by the original model of Shannon and it is shown there that all outputs of a GPAC satisfy
Equation \eqref{eq:ode}. The GPAC model of \cite{GC03} was further refined in
\cite{Gra04}, where a simpler structure of the GPAC circuits is shown to be equivalent to that presented in \cite{GC03}.

Here, we consider the
formal, nice and clear definition of \cite{GC03} of GPACs, and for
this class there is a clear equivalence between GPACs and polynomial
initial value problems of the form \eqref{eq:ode}. 
\end{remark}

We say that
a function
$f: \R \to \R$ is generable (by a GPAC) if and only if it corresponds to some component of a
solution of such a polynomial initial value problem \eqref{eq:ode}.

The discussion on how to go from univariate to multivariate
functions, that is to say from functions $f: \R \to \R^m$ to functions
$f: \R^n \to \R^m$ is briefly discussed in \cite{Sha41}, but no clear definitions and results
for this case have been stated or proved previously, up to our
knowledge. This is the purpose of the current paper. Another objective of this paper is to introduce
basic measures of the resources used by a GPAC (in particular on the
growth of functions), which might be used in the future to
establish complexity results for functions generated with GPACs.

We introduce the notion of \emph{generable} functions which are solutions of a polynomial
initial-value problem (PIVP) defined with an ODE \eqref{eq:ode}, and generalize this
notion to several input variables. We prove that this class enjoys a number of
stability and robustness properties.

Notice that extending the GPAC model to deal with several variables
have also been considered in \cite{PZ17}: Analog
networks on function data steams are considered, and
their semantic is obtained as the fixed point of suitable operators on
continuous data streams. A characterization of generable functions
generalizing some of Shannon's results is also provided.

The work we present here is different in the sense that we are interested in measuring
resources used in the GPAC and that we try to stay as close to the original GPAC as possible
(e.g.~we do not introduce new types of units or of data as done in \cite{PZ17}).

The paper is organized as follows:
\begin{itemize}
\item Section \ref{sec:generable} will introduce the notion of \emph{generable} function,
in the unidimensional and  multidimensional case. 
\item Section \ref{sec:pivp_stable} will give some stability properties of
  the class of generable functions, mostly
stability by arithmetic operations, composition and ODE solving.
\item Section \ref{sec:pivp_analytic} will show that generable functions are always analytic
\item Section \ref{sec:helper_func} will give a list of useful generable functions, as a way to see what can be achieved with
generable functions. In particular, we prove that a wide class of
functions (including piecewise defined functions, or periodic
functions) 
can be uniformly approximated over their domain. Notice that
Shannon proved a similar result but only over a compact domain (using
basically Weierstrass's theorem) and for dimension $1$. Here,
unlike Shannon, we prove a uniform approximation (distance can be 
controlled and set arbitrary small), and over the full domain of the
functions (not only over compact domains). 
\item Section \ref{sec:generable_field} will discuss the issue of
  constants. We give a few properties of \emph{generable fields}
which are fields with an extra property related to generable
functions, used in the previous proofs, and we prove basically that
constants can always be chosen to be polynomial time computable numbers. 
\end{itemize}

The current paper is mainly based on some extensions of results
present in the chapter 2 of the PhD document of Amaury
Pouly\footnote{\url{https://pastel.archives-ouvertes.fr/tel-01223284}}
  \cite{TheseAmaury}. This PhD was defended on July 2015, but
  presented results are original and have not been published
  otherwise. Furthermore, we go further here than what is established
  in Chapter 2 of this PhD document.
  
 Some results of this paper are already stated, without proofs in \cite{BGP16}, with a reference
 pointing to a preprint which ultimately would lead to the current paper. The difference between
 the two papers is that this paper focus on the class of \emph{generable functions} by GPACs, while
 \cite{BGP16} focus on the class of \emph{computable functions} by GPACs
 (see \cite{BCGH07} for an overview of the distinction between
 generable and computable functions by a GPAC. The work presented in this paper and in \cite{BGP16}
 extends earliers results present in \cite{BCGH07} by considering the multivariate case and complexity). The class of generable functions
 is discussed in detail here and many properties are proved (stability by several operations,
 analyticity, existence of a strict hierarchy of subclasses, etc.) and several functions and
 techniques which can be used for ``analog programming'' are introduced. We also consider the amount of resources used by a GPAC to perform several operations involving generable functions.
 
 The paper \cite{BGP16} by its turn focus on the class of computable functions by a GPAC, which are defined
 with the use of generable functions, hence the need to cite several results about
 generable functions which are proved here. In \cite{BGP16} we show that the class of computable functions
 is well defined (and that we can take into account bounded resources) and that several different (yet intuitive) notions of computability for the GPAC
 all yield the same class of functions, showing that a well-defined class of computable functions
 exists for the GPAC, even if we restrict the resources used by a GPAC.

\subsection{Notations}

In this paper, $\R$ denotes the real numbers, $\Rp=[0,+\infty)$ the nonnegative real numbers,
$\N=\{0,1,2,\ldots\}$ the natural numbers, $\Z$ the integers, $\intinterv{a}{b}=\{a,a+1,\ldots,b\}$
the integers between $a$ and $b$, $\Q$ the rational numbers,
$\Rpoly$ the polynomial time computable real numbers \cite{Ko91}, $\Rgen$ the
smallest generable field (see Section~\ref{sec:generable_field}). $\MAT{n}{d}{\K}$
denotes the set of $n\times d$ matrices over the ring $\K$. For any set $X$,
$\powerset{X}$ denotes the powerset of $X$ and $\card{X}$ the cardinal of $X$. For any function $f$, $\dom{f}$ is
the domain of $f$, $\fiter{f}{n}$ the $n^{th}$ iterate of $f$, $\frestrict{f}{X}$ the
restriction of $f$ to $X$, $\jacobian{f}(x)$ denotes the Jacobian matrix of $f$ at $x$.
For any vector $y\in\R^n$ and $e\leqslant n$, $y_{1..e}=(y_1,\ldots,y_e)$ denotes the first $e$ components of $y$
and $\infnorm{y}=\max(|y_1|,\ldots,|y_n|)$ denotes the infinity norm.
For any $x_0\in\R^n$ and $r>0$, $\openball{x_0}{r}=\{x:\twonorm{x-x_0}<r\}$ denotes
the open of radius $r$ and center $p$ for the euclidean norm. Given a (multivariate)
polynomial $p$, $\degp{p}$ denotes its degree and $\sigmap{p}$ the sum of the absolute
value of its coefficients. We denote by $\K[\R^d]$ the set of polynomial functions in $d$ variables
with coefficients in $\K$. Given a vector of polynomial $p=(p_1,\ldots,p_k)$,
which we simply refer to as a polynomial, $\degp{p}=\max(\degp{p_1},\ldots,\degp{p_k})$
and $\sigmap{p}=\max(\sigmap{p_1},\ldots,\sigmap{p_k})$. We denote by $\K^k[\R^d]$ the set of vectors
of polynomial functions in $d$ variables of size $k$ with coefficients in $\K$.
In this article, we write $\poly$ to denote an unspecified polynomial. For any $x\in\R$,
$\sgn{x}$ denotes the sign of $x$, $\lfloor x\rfloor$ the integer part of $x$,
$\intp_k(x)=\max(0,\min(k,\lfloor x\rfloor))$, $\round{x}$ the nearest integer (undefined
for $n+\tfrac{1}{2}$).

\section{Generable functions}\label{sec:generable}

In this section, we will define a notion of function generated by a
PIVP. From previous discussions, they correspond to functions
generated by the General Purpose Analog Computers of Claude Shannon
\cite{Sha41}; 

This class of functions is closed by a number of natural operations such
as arithmetic operators or composition. In particular, we will see that those functions are always analytic.
The major property of this class is the stability by ODE solving: if $f$ is \emph{generable}
and $y$ satisfies $y'=f(y)$ then $y$ is generable. This means that we can design
differential systems where the right-hand side contains much more general functions
than polynomials, and this system can be rewritten to use polynomials
only.

Several of the results here are extensions to the multidimensional
case of results established in \cite{TheseDaniel}. Moreover, 
a noticeable difference is that here we are also talking about complexity,
whereas \cite{TheseDaniel} is often not precise about the growth of
functions as only motivated by computability theory.

In this section, $\K$ will always refer to a real field, for example $\K=\Q$.
The basic definitions work for any such field but the main results will require
some assumptions on $\K$. These assumptions will be formalized in
Definition \ref{def:generable_field}
and detailed in Section \ref{sec:generable_field}.

\subsection{Unidimensional case}

We start with the definition of generable functions from $\R$ to
$\R^n$.  Those are defined as the solution of some polynomial IVP
(PIVP) with an additional boundedness constraint. This will be of
course key to talk about complexity theory for the GPAC, since if no constraint is
put on the growth of functions, it is easy to see that arbitrary
growing functions can be generated by a GPAC (or, equivalently, by a PIVP), such as the
$t \mapsto \exp(\exp(\dots \exp(t)))$ function. Indeed consider the following system
\[
\left\{\begin{array}{@{}r@{}l}y_1(0)&=1\\y_2(0)&=1\\\ldots&\\ y_n(0)&=1
\end{array}\right.
\qquad
\left\{\begin{array}{@{}r@{}l}y_1'(t)&=y_1(t)\\y_2(t)&=y_1(t)y_2(t)\\\ldots&\\
y_d'(t)&=y_1(t)\cdots y_n(t)
\end{array}\right.
\]
This system has the form \eqref{eq:ode} and can be solved explicitly. It has the following solution:
\[y_1(t)=e^t\qquad y_{n+1}(t)=e^{y_n(t)-1}
\qquad y_d(t)=e^{e^{\iddots^{e^{e^t}-1}}-1}\label{eq:exp_tower}
\]

Hence, although previous papers about the GPAC studied computability,
like \cite{Sha41}, \cite{Pou74}, \cite{GC03} or \cite{Gra04},
they said nothing about complexity. And as the previous example shows,
the output of a GPAC can have an arbitrarily high growth and thus arbitrarily high complexity.
Hence, to distinguish between reasonable GPACs, it is natural to bound the growth of
the outputs of a GPAC and use those bounds as a complexity measure. Moreover, as we
have shown in \cite{BGP12}, we can compute (in the Computable Analysis setting
\cite{newcomputationalparadigms}) the solution of a PIVP in time polynomial in the growth
bound of the PIVP. This motivates the following definition (in what follows, $\K[\R^n]$ denotes polynomial functions
with $n$ variables and with coefficients in $\K$, where variables live
in $\R^n$ and\footnote{We write $[a,b]$ (respectively:
  $]a,b]$, $[a,b[$, $]a,b[$) for closed
  (resp. semi-closed, open)  interval.} $\Rp=[0,+\infty [$): 

\begin{definition}[Generable function]\label{def:gpac_generable}
Let $\mtt{sp}:\Rp\rightarrow\Rp$ be a nondecreasing function and $f:\R\rightarrow\R^m$.
We say that $f\in\gval[\K]{\mtt{sp}}$ if and only
if there exists $n\geqslant m$, $y_0\in\K^n$ and $p\in\K^n[\R^n]$ such that there is a
(unique) $y:\R\rightarrow\R^n$ satisfying for all time $t\in\R$:
\begin{itemize}
\item $y'(t)=p(y(t))$ and $y(0)=y_0$\hfill$\blacktriangleright$ $y$ satisfies a differential equation
\item $f(t)=y_{1..m}(t)=(y_1(t),\ldots,y_m(t))$\hfill$\blacktriangleright$ $f$ is a component of $y$
\item $\infnorm{y(t)}\leqslant \mtt{sp}(|t|)$\hfill$\blacktriangleright$ $y$ is bounded by $\mtt{sp}$
\end{itemize}
The set of all generable functions is denoted by
$\gval[\K]{}=\bigcup_{\mtt{sp}:\R\rightarrow\Rp}\gval[\K]{\mtt{sp}}$.
When this is not ambiguous, we do not specify the field $\K$ and write $\gval{\mtt{sp}}$
or simply $\gval{}$. We will also write $\gval{\poly}$ (or $\gval[\K]{\poly}$)  as a synonym of
$\gval{\mtt{sp}}$ (respectively: $\gval[\K]{\mtt{sp}}$) for some
polynomial $\mtt{sp}$ (see coming Remark \ref{rq:vingtquatre}). 
\end{definition}

\begin{remark}[Uniqueness]\label{rem:gpac_uniqueness_regularity}
The uniqueness of $y$ in Definition \ref{def:gpac_generable} is a consequence of the
Cauchy-Lipschitz theorem. Indeed a polynomial is a locally Lipschitz function.
\end{remark}

\begin{remark}[Regularity]\label{rem:gpac_regularity}
As a consequence of the Cauchy-Lipschitz theorem, the solution $y$ in
Definition \ref{def:gpac_generable} is at least $C^\infty$. It can be seen that
it is in fact real analytic, as it is the case for analytic differential equations
in general \cite{Arn78}.
\end{remark}

\begin{remark}[Multidimensional output]\label{rem:multidim_out}
It should be noted that although Definition \ref{def:gpac_generable} defines generable
functions with output in $\R^m$, it is completely equivalent to say that $f$ is
generable if and only if each of its component is (\emph{i.e.} $f_i$ is generable for every $i$);
and restrict the previous definition to functions from $\R$ to $\R$ only.
Also note that if $y$ is the solution from Definition \ref{def:gpac_generable},
then obviously $y$ is generable.
\end{remark}

Although this might not be obvious at first glance, this class contains polynomials,
and contains many elementary functions such as the exponential function, as well
as the trigonometric functions. Intuitively, all functions in this class can be
computed efficiently by classical machines, where $\mtt{sp}$ measures
some ``hardness'' in computing the function.  We took care to choose the constants
such as the initial time and value, and the coefficients of the polynomial in $\K$.
The idea is to prevent any uncomputability from arising by the choice of
uncomputable real numbers in the constants.

\begin{example}[Polynomials are generable]\label{ex:poly_generable}
Let $p$ in $\Q(\pi)[\R]$. For example $p(x)=x^7-14x^3+\pi^2$. We will show that $p\in\gval[\K]{\mtt{sp}}$
where $\mtt{sp}(x)=x^7+14x^3+\pi^2$. We need to rewrite $p$ with a polynomial differential equation:
we immediately get that $p(0)=\pi^2$ and $p'(x)=7x^6-42x^2$.
However, we cannot express $p'(x)$ as a polynomial of $p(x)$ only: we need access
to $x$. This can be done by introducing a new variable $v(x)$ such that $v(x)=x$.
Indeed, $v'(x)=1$ and $v(0)=0$. Finally we get:
\[\left\{\begin{array}{@{}r@{}l}p(0)&=\pi^2\\p'(x)&=7v(x)^6-42v(x)^2\end{array}\right.
\qquad
\left\{\begin{array}{@{}r@{}l}v(0)&=0\\v'(x)&=1\end{array}\right.
\]
Formally, we define $y(x)=(p(x),x)$ and show that $y(0)=(\pi^2,0)\in\K^2$ and $y'(x)=p(y(x))$
where $p_1(a,b)=7b^6-42b^2$ and $p_2(a,b)=1$. Also note that the coefficients are clearly in $\Q(\pi)$).
We also need to check that $\mtt{sp}$ is a bound on $\infnorm{y(x)}$ (for $x\geq 0$):
\[\infnorm{y(x)}=\max(|x|,|x^7-14x^3+\pi^2|)\leq\mtt{sp}(x)\]
This shows that $p\in\gval[\K]{\mtt{sp}}$ and can be generalized to show that any polynomial
in one variable is generable.
\end{example}

\begin{example}[Some generable elementary functions]\label{ex:elem_func_gpac}
We will check that $\exp\in\gval[\Q]{\exp}$ and $\sin,\cos,\tanh\in\gval[\Q]{\lambdafunex{x}{1}}$.
We will also check that $\arctan\in\gval[\Q]{\lambdafunex{x}{\max(x,\frac{\pi}{2})}}$.
\begin{itemize}
\item A characterization of the exponential function is the following: $\exp(0)=1$ and $\exp'=\exp$.
Since $\infnorm{\exp}=\exp$, it is immediate that $\exp\in\gval[\Q]{\exp}$.
The exponential function might be the simplest generable function.
\item The sine and cosine functions are related by their derivatives since $\sin'=\cos$ and
$\cos'=-\sin$. Also $\sin(0)=0$ and $\cos(0)=1$, and $\infnorm{(\sin(x),\cos(x))}\leqslant 1$,
we get that  $\sin,\cos\in\gval[\Q]{\lambdafunex{x}{1}}$ with the same system.
\item The hyperbolic tangent function will be very useful in this paper. Is it known to satisfy the very simple
polynomial differential equation $\tanh'=1-\tanh^2$. Since $\tanh(0)=0$ and $|\tanh(x)|\leqslant1$,
this shows that $\tanh\in\gval[\Q]{\lambdafunex{x}{1}}$.
\item Another very useful function will be the arctangent function. A possible definition of the arctangent
is the unique function satisfying $\arctan(0)=0$ and $\arctan'(x)=\frac{1}{1+x^2}$.
Unfortunately this is neither a polynomial in $\arctan(x)$ nor in $x$. A common trick is
to introduce a new variable $z(x)=\frac{1}{1+x^2}$ so that $\arctan'(x)=z(x)$, in the hope
that $z$ satisfies a PIVP. This is the case since $z(0)=1$ and $z'(x)=\frac{-2x}{(1+x^2)^2}=-2xz(x)^2$
which is a polynomial in $z$ and $x$. We introduce a new variable for $x$ as we did in the previous examples.
Finally, define $y(x)=(\arctan(x),\frac{1}{1+x^2},x)$
and check that $y(0)=(0,1,0)$ and $y'(x)=(y_2(x),-2y_3(x)y_2(x)^2,1)$. The $\frac{\pi}{2}$
bound on $\arctan$ is a textbook property, and the bound on the other variables is immediate.
\end{itemize}
\end{example}

Not only the class of generable functions contains many classical and useful functions,
but it is also closed under many operations. We will see that the sum, difference, product and composition
of generable functions are still generable. Before moving on to the properties of this class,
we need to mention the easily overlooked issue about constants, best illustrated as
an example.

\begin{example}[The issue of constants]
Let $\K$ be a field, containing at least the rational numbers. Assume
that generable functions are closed under composition, that is for any two $f,g\in\gval[\K]{}$
we have $f\circ g\in\gval[\K]{}$. Let $\alpha\in\K$ and $g=\lambdafunex{x}{\alpha}$.
Then for any $(f:\R\rightarrow\R)\in\gval[\K]{}$,
$f\circ g\in\gval[\K]{}$. Using Definition \ref{def:gpac_generable}, we
get that $f(g(0))\in\K$ which means $f(\alpha)\in\K$ for any $\alpha\in\K$.
In other words, $\K$ must satisfy the following property:
\[f(\K)\subseteq\K\qquad\forall f\in\gval[\K]{}\]
This property does not hold for general fields.
\end{example}

The example above outlines the need for a stronger hypothesis on $\K$ if
we want to be able to compose functions. Motivated by this example, we introduce the
following notion of \emph{generable field}.

\begin{definition}[Generable field]\label{def:generable_field}
A field $\K$ is \emph{generable} if and only if $\Q\subseteq\K$ and for any
$\alpha\in\K$ and $(f:\R\rightarrow\R)\in\gval[\K]{}$, we have $f(\alpha)\in\K$.
\end{definition}

\begin{tikzpicture}[baseline=-.75ex] \node[shape=regular polygon, regular polygon sides=3, inner sep=0pt, draw, thick] {\textbf{!}};\end{tikzpicture}%
\hfill%
\begin{minipage}{0.9\linewidth}
From now on, we will assume that $\K$ is a generable field. See Section \ref{sec:generable_field}
for more details on this assumption.
\end{minipage}

\begin{example}[Usual constants are generable]
In this paper, we will use again and again that some well-known constants
belong to any generable field. We detail the proof for $\pi$ and $e$:
\begin{itemize}
\item It is well-known that $\frac{\pi}{4}=\arctan(1)$. We saw in Example \ref{ex:elem_func_gpac}
    that $\arctan\in\gval[\Q]{}$ and since $1\in\K$ we get that $\frac{\pi}{4}\in\K$ because
    $\K$ is a generable field. We conclude that $\pi\in\K$ because $\K$ is a field and $4\in\K$.
\item By definition, $e=\exp(1)$ and $\exp\in\gval[\Q]{}$, so $e\in\K$ because $\K$
    is a generable field and $1\in\K$.
\end{itemize}
\end{example}

\begin{lemma}[Arithmetic on generable functions]\label{lem:gpac_gen_op}
Let $f\in\gval{\mtt{sp}}$ and $g\in\gval{\mtt{\ovl{sp}}}$.
\begin{itemize}
\item $f+g, f-g\in\gval{\mtt{sp}+\ovl{\mtt{sp}}}$
\item $fg\in\gval{\max(\mtt{sp},\ovl{\mtt{sp}},\mtt{sp}\thinspace\ovl{\mtt{sp}})}$
\item $\frac{1}{f}\in\gval{\max(\mtt{sp},\mtt{sp}')}$ where $\mtt{sp}'(t)=\frac{1}{|f(t)|}$, if $f$ never cancels
\item $f\circ g\in\gval{\max(\ovl{\mtt{sp}},\mtt{sp}\circ\ovl{\mtt{sp}})}$
\end{itemize}
Note that the first three items only require that $\K$ is a field, whereas the
last item also requires $\K$ to be a generable field.
\end{lemma}

\begin{proof}
Assume that $f:\R\rightarrow\R^m$ and $g:\R\rightarrow\R^\ell$.
We will make a detailed proof of the product and composition cases, since the
sum and difference are much simpler. The intuition follows from basic differential
calculus and the chain rule: $(fg)'=f'g+fg'$ and $(f\circ g)'=g'(f'\circ g)$.
Note that $\ell=1$ for the composition to make sense and $\ell=m$ for the product to make
sense (componentwise). The only difficulty
in this proof is technical: the differential equation may include more variables
than just the ones computing $f$ and $g$. This requires a bit of notation to
stay formal.
Apply Definition \ref{def:gpac_generable} to $f$ and $g$ to get $p,\ovl{p},y_0,\ovl{y}_0$.
Consider the following systems:
\[
\left\{\begin{array}{r@{}l}y(0)&=y_0\\y'(t)&=p(y(t))\\
\ovl{y}(0)&=\ovl{y}_0\\\ovl{y}'(t)&=\ovl{p}(\ovl{y}(t))\end{array}\right.
\qquad\left\{\begin{array}{r@{}l}z_i(0)&=y_{0,i}\ovl{y}_{0,i}\\
z_i'(t)&=p_i(y(t))\bar{y}_i(t)+y_i(t)\ovl{p}_i(\bar{y}(t))\\
u_i(0)&=f_i(\ovl{y}_{0,1})\\
u_i'(t)&=\ovl{p}_i(\ovl{y}(t))p(u(t))\end{array}\right.
\qquad i\in\intinterv{1}{m}
\]
Those systems are clearly polynomial. By construction, $u$ and $z$ exist over $\R$
since $z_i(t)=y_i(t)\bar{y}_i(t)$ satisfies the differential equation over
$\R$ (indeed $y$ and $\bar{y}$ exist over $\R$). Similarly, $u_i(t)=y_i(\bar{y}(t))$
exists over $\R$ and satisfies the equation.
Remember that by definition, for any $i\in\intinterv{1}{m}$ and $j\in\intinterv{1}{\ell}$,
$f_i(t)=y_i(t)$ and $g_j(t)=z_j(t)$. Consequently, $z_i(t)=f_i(t)g_i(t)$
and $u_i(t)=f_i(g_1(t))$.

Also by definition, $\infnorm{y(t)}\leqslant\mtt{sp(t)}$ and $\infnorm{\ovl{y}(t)}\leqslant\mtt{\ovl{sp}}(t)$.
It follows that $|z_i(t)|\leqslant|y_i(t)||\ovl{y}_i(t)|\leqslant\mtt{sp}(t)\mtt{\ovl{sp}}(t)$,
and similarly $|u_i(t)|\leqslant |f_i(g_1(t))|\leqslant\mtt{sp}(g_1(t))\leqslant\mtt{sp}(\mtt{\ovl{sp}}(t))$.

The case of $\frac{1}{g}$ is very similar: define $g=\frac{1}{f}$ then $g'=-f'g^2$.
The only difference is that we don't have an a priori bound on $g$ except $\frac{1}{|f|}$,
and we must assume that $f$ is never zero for $g$ to be defined over $\R$.

Finally, a very important note about constants and coefficients which appear in those systems.
It is clear that $y_{0,i}\ovl{y}_{0,i}\in\K$ because $\K$ is a field. Similarly,
for $\frac{1}{f}$ we have $\frac{1}{f(0)}=\frac{1}{y_{0,1}}\in\K$. However,
there is no reason in general for $f_i(\ovl{y}_{0,1})$ to belong to $\K$,
and this is where we need the assumption that $\K$ is generable.
\end{proof}

\subsection{Multidimensional case}

We introduced generable functions as a special kind of function from $\R$ to
$\R^n$. We saw that this class nicely contains polynomials, however it comes
with two defects which prevents other interesting functions from being generable:
\begin{itemize}
\item The domain of definition is $\R$: this is very strong, since other ``easy''
  targets such as $\tan$, $\log$ or even $\lambdafunex{x}{\frac{1}{x}}$ cannot be defined,
  despite satisfying polynomial differential equations.
\item The domain of definition is one-dimensional: it would be useful to define
  generable functions in several variables, like multivariate polynomials.
\end{itemize}
The first issue can be dealt with by adding restrictions on the domain where the
differential equation holds, and by shifting the initial condition
($0$ might not belong to the domain). Overcoming the second problem is
less obvious.

The examples below give two intuitions before introducing the formal definition.
The first example draws inspiration from multivariate calculus and differential
form theory. The second example focuses on GPAC composition. As we will see,
both examples highlight the same properties of multidimensional generable functions.

\begin{figure}
\begin{minipage}{0.45\linewidth}%
\begin{center}
\begin{tikzpicture}[scale=1.5]
 \begin{scope}[shift={(0,0)},rotate=0]
  \draw (0,0) -- (0.7,0) -- (0.7,0.7) -- (0,0.7) -- (0,0);
  \node at (.35,.35) {$\int$};
 \end{scope}
 \draw (.7,.35) -- (1.4,.35);
 \node[anchor=west] at (1.4,.35) {$f(t)=e^t$};
 \draw (1,.35) -- (1,1) -- (-.3,1) -- (-.3,.525) -- (0,.525);
 \fill (1,.35) circle[radius=.07];
 \draw (-.7,.175) -- (0,.175);
 \node[anchor=east] at (-.7,.175) {$t$};
\end{tikzpicture}
\end{center}
\caption{Simple GPAC}
\label{fig:multidim_gpac_simple_int}
\end{minipage}%
\hfill
\begin{minipage}{0.60\linewidth}%
\begin{center}
\begin{tikzpicture}[scale=1.5]
  \begin{scope}[shift={(0,0)},rotate=0]
    \begin{scope}[shift={(0,0)},rotate=0]
    \draw (0,0) -- (0.7,0) -- (0.7,0.7) -- (0,0.7) -- (0,0);
    \node at (.35,.35) {$\int$};
    \end{scope}
    \draw (.7,.35) -- (1.4,.35);
    \draw (-.7,.175) -- (0,.175);
    \node[anchor=east] at (-.7,.175) {$x_1$};
  \end{scope}
  \begin{scope}[shift={(0,1)},rotate=0]
    \begin{scope}[shift={(0,0)},rotate=0]
    \draw (0,0) -- (0.7,0) -- (0.7,0.7) -- (0,0.7) -- (0,0);
    \node at (.35,.35) {$\int$};
    \end{scope}
    \draw (.7,.35) -- (1.4,.35);
    \draw (-.7,.175) -- (0,.175);
    \node[anchor=east] at (-.7,.175) {$x_2$};
  \end{scope}
  \begin{scope}[shift={(2,0.5)}]
    \draw (0,0) -- (0.7,0) -- (0.7,0.7) -- (0,0.7) -- (0,0);
    \node at (.35,.35) {$+$};
  \end{scope}
  \begin{scope}[shift={(-2.1,0.5)}]
    \draw (0,0) -- (0.7,0) -- (0.7,0.7) -- (0,0.7) -- (0,0);
    \node at (.35,.35) {$1$};
    \draw (.7,.35) -- (1.4,.35);
  \end{scope}
  \draw (-0.7,0.85) -- (-0.35,0.85) -- (-0.35,1.525) -- (0,1.525);
  \fill (-0.35,0.85) circle[radius=.07];
  \draw (-0.35,0.85) -- (-0.35,0.525) -- (0,0.525);
  \draw (1.4,.35) -- (1.4,0.73) -- (2,0.73);
  \draw (1.4,1.35) -- (1.4,0.96) -- (2,0.96);
  \draw (2.7,0.85) -- (3.1,0.85);
  \node[anchor=west] at (3.1,0.85) {$g$};
\end{tikzpicture}
\end{center}
\caption{GPAC with two inputs}
\label{fig:multidim_gpac_simple}
\end{minipage}%
\end{figure}

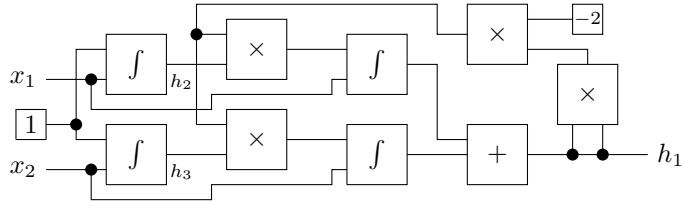
\begin{figure}
\begin{center}
\begin{tikzpicture}[scale=0.8]
 \begin{scope}[shift={(10,0)},rotate=0]
  \draw (0,0) -- (1,0) -- (1,1) -- (0,1) -- (0,0);
  \node at (.5,.5) {$+$};
  \draw (1,.5) -- (3,.5);
  \node[anchor=west] at (3,.5) {$h_1$};
 \end{scope}
 \begin{scope}[shift={(8,0)},rotate=0]
  \draw (0,0) -- (1,0) -- (1,1) -- (0,1) -- (0,0);
  \node at (.5,.5) {$\int$};
  \draw (1,.5) -- (2,.5);
 \end{scope}
 \begin{scope}[shift={(8,1.5)},rotate=0]
  \draw (0,0) -- (1,0) -- (1,1) -- (0,1) -- (0,0);
  \node at (.5,.5) {$\int$};
  \draw (1,.5) -- (1.5,.5) -- (1.5,-0.75) -- (2,-0.75);
 \end{scope}
 \begin{scope}[shift={(11.5,1)},rotate=0]
  \draw (0,0) -- (1,0) -- (1,1) -- (0,1) -- (0,0);
  \node at (.5,.5) {$\times$};
  \draw (0.25, 0) -- (0.25,-.5);
  \draw (0.75, 0) -- (0.75,-.5);
  \fill (0.25,-.5) circle[radius=.1];
  \fill (0.75,-.5) circle[radius=.1];
  \draw (0.5, 1) -- (.5, 1.25) -- (-.5,1.25);
 \end{scope}
 \begin{scope}[shift={(11.75,2.5)},rotate=0]
  \draw (0,0) -- (.5,0) -- (.5,.5) -- (0,.5) -- (0,0);
  \node at (.25,.25) {$\scriptstyle -2$};
  \draw (0,.25) -- (-.75,.25);
 \end{scope}
 \begin{scope}[shift={(10,2)},rotate=0]
  \draw (0,0) -- (1,0) -- (1,1) -- (0,1) -- (0,0);
  \node at (.5,.5) {$\times$};
  \draw (0,.5) -- (-.5,.5) -- (-.5,1) -- (-4.5,1) -- (-4.5,-1);
 \end{scope}
 \begin{scope}[shift={(6,1.75)},rotate=0]
  \draw (0,0) -- (1,0) -- (1,1) -- (0,1) -- (0,0);
  \node at (.5,.5) {$\times$};
  \draw (1,.5) -- (2,.5);
  \draw (0,.75) -- (-.5,.75);
  \fill (-.5,.75) circle[radius=.1];
 \end{scope}
 \begin{scope}[shift={(4,0)},rotate=0]
  \draw (0,0) -- (1,0) -- (1,1) -- (0,1) -- (0,0);
  \node at (.5,.5) {$\int$};
  \draw (1,.5) -- (2,.5);
  \draw (0,.25) -- (-1,.25);
  \node[anchor=east] at (-1,.25) {$x_2$};
  \node [anchor=north] at (1.25,.5) {$\scriptstyle h_3$};
 \end{scope}
 \begin{scope}[shift={(6,0.25)},rotate=0]
  \draw (0,0) -- (1,0) -- (1,1) -- (0,1) -- (0,0);
  \node at (.5,.5) {$\times$};
  \draw (1,.5) -- (2,.5);
  \draw (0,.75) -- (-.5,.75);
 \end{scope}
 \begin{scope}[shift={(4,1.5)},rotate=0]
  \draw (0,0) -- (1,0) -- (1,1) -- (0,1) -- (0,0);
  \node at (.5,.5) {$\int$};
  \draw (1,.5) -- (2,.5);
  \draw (0,.25) -- (-1,.25);
  \node[anchor=east] at (-1,.25) {$x_1$};
  \node [anchor=north] at (1.25,.5) {$\scriptstyle h_2$};
 \end{scope}
 \begin{scope}[shift={(2.5,.75)},rotate=0]
  \draw (0,0) -- (.5,0) -- (.5,.5) -- (0,.5) -- (0,0);
  \node at (.25,.25) {$1$};
  \draw (.5,.25) -- (1,.25) -- (1,1.5) -- (1.5,1.5);
  \draw (1,.25) -- (1,0) -- (1.5,0);
  \fill (1,.25) circle[radius=.1];
 \end{scope}
 \begin{scope}[shift={(3.75, 1.75)},rotate=0]
  \draw (0,0) -- (0, -.5) -- (2, -.5) -- (2,-.25) -- (4,-.25) -- (4,0) -- (4.25,0);
  \fill (0,0) circle[radius=.1];
 \end{scope}
 \begin{scope}[shift={(3.75, .25)},rotate=0]
  \draw (0,0) -- (0, -.5) -- (2, -.5) -- (2,-.25) -- (4,-.25) -- (4,0) -- (4.25,0);
  \fill (0,0) circle[radius=.1];
 \end{scope}
\end{tikzpicture}
\end{center}
\caption{A more involved multidimensional GPAC}
\label{fig:multidim_gpac_complex}
\end{figure}

\begin{figure}
\begin{center}
\begin{tikzpicture}[scale=0.8]
 \begin{scope}[shift={(0,0)}]
  \begin{scope}[shift={(0,0)}]
   \begin{scope}[shift={(0, 0)},rotate=0]
    \draw (0,0) -- (1,0) -- (1,1) -- (0,1) -- (0,0);
    \node at (.5,.5) {$\int$};
    \draw (1,.5) -- (1.5,.5) -- (1.5, .25) -- (2, .25);
    \draw (0,.25) -- (-1,.25);
    \node[anchor=east] at (-1,.25) {$t$};
    \draw (0,.75) -- (-1,.75);
    \node[anchor=east] at (-1,.75) {$u$};
   \end{scope}
   \begin{scope}[shift={(2,0)},rotate=0]
    \draw (0,0) -- (1,0) -- (1,1) -- (0,1) -- (0,0);
    \node at (.5,.5) {$\int$};
    \draw (1,.5) -- (1.5,.5);
    \node[anchor=west] at (1.5,.5) {$w$};
    \draw (0,.75) -- (-.5,.75);
    \node[anchor=south] at (-.5,.75) {$v$};
   \end{scope}
  \end{scope}
  \node at (5, .5) {$\leadsto$};
  \begin{scope}[shift={(8,0)}]
   \begin{scope}[shift={(0, 0)},rotate=0]
    \draw (0,0) -- (1,0) -- (1,1) -- (0,1) -- (0,0);
    \node at (.5,.5) {$\times$};
    \draw (1,.5) -- (1.5,.5) -- (1.5, .75) -- (2, .75);
    \draw (0,.25) -- (-1,.25);
    \node[anchor=east] at (-1,.25) {$u$};
    \draw (0,.75) -- (-1,.75);
    \node[anchor=east] at (-1,.75) {$v$};
   \end{scope}
   \begin{scope}[shift={(2,0)},rotate=0]
    \draw (0,0) -- (1,0) -- (1,1) -- (0,1) -- (0,0);
    \node at (.5,.5) {$\int$};
    \draw (1,.5) -- (1.5,.5);
    \node[anchor=west] at (1.5,.5) {$w$};
    \draw (0,.25) -- (-.5,.25);
    \node[anchor=north] at (-.5,.25) {$t$};
   \end{scope}
  \end{scope}
 \end{scope}
 \begin{scope}[shift={(0,-3)}]
  \begin{scope}[shift={(0,0)}]
   \begin{scope}[shift={(0, 0)},rotate=0]
    \draw (0,0) -- (1,0) -- (1,1) -- (0,1) -- (0,0);
    \node at (.5,.5) {$\times$};
    \draw (1,.5) -- (1.5,.5) -- (1.5, .25) -- (2, .25);
    \draw (0,.25) -- (-1,.25);
    \node[anchor=east] at (-1,.25) {$v$};
    \draw (0,.75) -- (-1,.75);
    \node[anchor=east] at (-1,.75) {$u$};
   \end{scope}
   \begin{scope}[shift={(2,0)},rotate=0]
    \draw (0,0) -- (1,0) -- (1,1) -- (0,1) -- (0,0);
    \node at (.5,.5) {$\int$};
    \draw (1,.5) -- (1.5,.5);
    \node[anchor=west] at (1.5,.5) {$y$};
    \draw (0,.75) -- (-.5,.75);
    \node[anchor=south] at (-.5,.75) {$w$};
   \end{scope}
  \end{scope}
  \node at (5, .5) {$\leadsto$};
  \begin{scope}[shift={(11,0)}]
   \begin{scope}[shift={(-2, 1)},rotate=0]
    \draw (0,0) -- (1,0) -- (1,1) -- (0,1) -- (0,0);
    \node at (.5,.5) {$\times$};
    \draw (1,.5) -- (1.5,.5) -- (1.5,.75) -- (2,.75);
   \end{scope}
   \begin{scope}[shift={(0, 1)},rotate=0]
    \draw (0,0) -- (1,0) -- (1,1) -- (0,1) -- (0,0);
    \node at (.5,.5) {$\int$};
    \draw (1,.5) -- (1.5,.5) -- (1.5, -.25) -- (2, -.25);
   \end{scope}
   \begin{scope}[shift={(-2, -1)},rotate=0]
    \draw (0,0) -- (1,0) -- (1,1) -- (0,1) -- (0,0);
    \node at (.5,.5) {$\times$};
    \draw (1,.5) -- (1.33,.5) -- (1.33,.75) -- (2,.75);
   \end{scope}
   \begin{scope}[shift={(0, -1)},rotate=0]
    \draw (0,0) -- (1,0) -- (1,1) -- (0,1) -- (0,0);
    \node at (.5,.5) {$\int$};
    \draw (1,.5) -- (1.5,.5) -- (1.5, 1.25) -- (2, 1.25);
   \end{scope}
   \begin{scope}[shift={(2,0)},rotate=0]
    \draw (0,0) -- (1,0) -- (1,1) -- (0,1) -- (0,0);
    \node at (.5,.5) {$+$};
    \draw (1,.5) -- (1.5,.5);
    \node[anchor=west] at (1.5,.5) {$y$};
   \end{scope}
   \draw (-4, 1.25) -- (-2,1.25);
   \draw (-2.5, 1.25) -- (-2.5, -.25) -- (-2, -.25);
   \fill (-2.5, 1.25) circle[radius=0.1];
   \node[anchor=east] at (-4, 1.25) {$w$};
   \draw (-4,1.75) -- (-2,1.75);
   \node[anchor=east] at (-4, 1.75) {$u$};
   \draw (-4,-.75) -- (-2,-.75);
   \node[anchor=east] at (-4, -.75) {$v$};
   \draw (-3,-.75) -- (-3,.75) -- (-.5,.75) -- (-.5,1.25) -- (0,1.25);
   \fill (-3,-.75) circle[radius=0.1];
   \draw (-3.5,1.75) -- (-3.5,.25) -- (-.33,.25) -- (-.33,-.75) -- (0,-.75);
   \fill (-3.5,1.75) circle[radius=0.1];
  \end{scope}
 \end{scope}
 \begin{scope}[shift={(0,-6.5)}]
  \begin{scope}[shift={(0,0)}]
   \begin{scope}[shift={(0, 0)},rotate=0]
    \draw (0,0) -- (1,0) -- (1,1) -- (0,1) -- (0,0);
    \node at (.5,.5) {$+$};
    \draw (1,.5) -- (1.5,.5) -- (1.5, .25) -- (2, .25);
    \draw (0,.25) -- (-1,.25);
    \node[anchor=east] at (-1,.25) {$v$};
    \draw (0,.75) -- (-1,.75);
    \node[anchor=east] at (-1,.75) {$u$};
   \end{scope}
   \begin{scope}[shift={(2,0)},rotate=0]
    \draw (0,0) -- (1,0) -- (1,1) -- (0,1) -- (0,0);
    \node at (.5,.5) {$\int$};
    \draw (1,.5) -- (1.5,.5);
    \node[anchor=west] at (1.5,.5) {$y$};
    \draw (0,.75) -- (-.5,.75);
    \node[anchor=south] at (-.5,.75) {$w$};
   \end{scope}
  \end{scope}
  \node at (5, .5) {$\leadsto$};
  \begin{scope}[shift={(9,0)}]
   \begin{scope}[shift={(0, .75)},rotate=0]
    \draw (0,0) -- (1,0) -- (1,1) -- (0,1) -- (0,0);
    \node at (.5,.5) {$\int$};
    \draw (1,.5) -- (1.5,.5) -- (1.5, 0) -- (2, 0);
    \draw (-1,.75) -- (0,.75);
    \draw (-.5,.75) -- (-.5,.75) -- (-.5,-.75) -- (0,-.75);
    \fill (-.5,.75) circle[radius=.1];
    \node[anchor=east] at (-1,.75) {$w$};
    \draw (-1,.25) --(0,.25);
    \node[anchor=east] at (-1,.25) {$u$};
   \end{scope}
   \begin{scope}[shift={(0, -.75)},rotate=0]
    \draw (0,0) -- (1,0) -- (1,1) -- (0,1) -- (0,0);
    \node at (.5,.5) {$\int$};
    \draw (1,.5) -- (1.5,.5) -- (1.5, 1) -- (2, 1);
    \draw (-1,.25) --(0,.25);
    \node[anchor=east] at (-1,.25) {$v$};
   \end{scope}
   \begin{scope}[shift={(2,0)},rotate=0]
    \draw (0,0) -- (1,0) -- (1,1) -- (0,1) -- (0,0);
    \node at (.5,.5) {$+$};
    \draw (1,.5) -- (1.5,.5);
    \node[anchor=west] at (1.5,.5) {$y$};
   \end{scope}
  \end{scope}
 \end{scope}
\end{tikzpicture}
\end{center}
\caption{GPAC rewriting}
\label{fig:multidim_gpac_rewriting}
\end{figure}
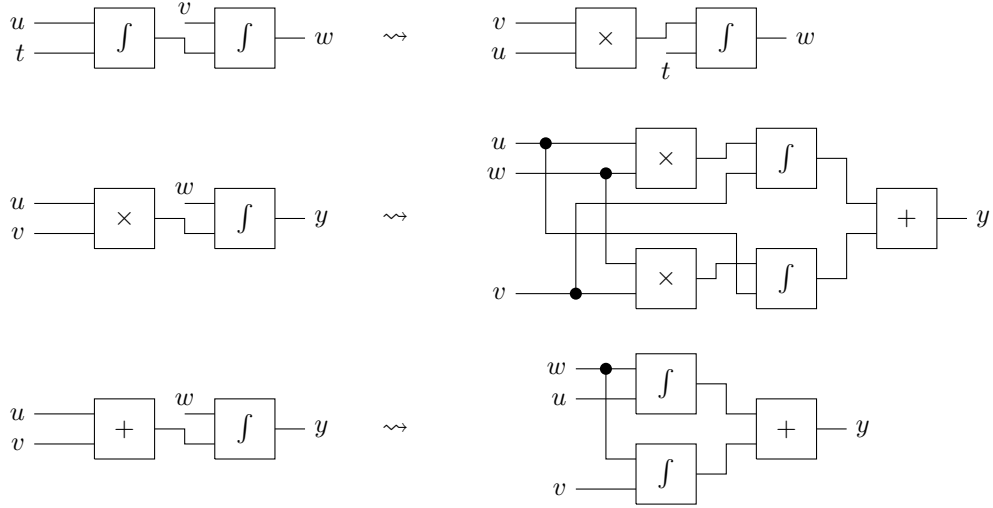

\begin{example}[Multidimensional GPAC]
The history and motivation for the GPAC have been described above. The GPAC
is the starting point for the definition of generable functions. It crucially
relies on the integrator unit to build interesting circuits. In modern terms,
the integration is often done implicitly with respect to time, as shown in
\figref{fig:multidim_gpac_simple_int} where the corresponding equation
is $f(t)=\int f$, or $f'=f$. Notice that the circuit has a single
``floating input'' which is $t$ and is only used in the ``derivative port''
of the integrator. What would be the meaning of a circuit with several
such inputs, as shown in \figref{fig:multidim_gpac_simple} ? Formally writing the system
and differentiating gives:
\begin{align*}
g&=\int 1dx_1+\int 1dx_2=x_1+x_2\\
dg&=dx_1+dx_2
\end{align*}
\figref{fig:multidim_gpac_complex} gives a more interesting example to better grasp
the features of these GPAC. Using the same ``trick'' as before we get:
\[\begin{array}{cc}
\begin{array}{rl}
h_2&=\int 1dx_1\\
h_3&=\int 1dx_2\\
h_1&=\int -2h_1^2h_2dx_1+\int-2h_1^2h_3dx_2
\end{array}&
\begin{array}{rl}
dh_2&=dx_1\\
dh_3&=dx_2\\
dh_1&=-2h_1^2h_2dx_1-2h_1^2h_3dx_2
\end{array}
\end{array}
\]

It is now apparent that the computed function $h$ satisfies a special
property because $dh_1(x)=p_1(h_1,h_2,h_3)dx_1+p_2(h_1,h_2,h_3)dx_2$ where $p_1$ and $p_2$
are polynomials. In other words, $dh_1=\scalarprod{p(h)}{dx}$ where $h=(h_1,h_2,h_3)$,
$x=(x_1,x_2)$ and $p=(p_1,p_2)$ is a polynomial vector. We obtain similar equations for
$h_2$ and $h_3$. Finally,
$dh=q(h)dx$ where $q(h)$ is the polynomial matrix given by:
\[q(h)=\begin{pmatrix}-2h_1^2h_2&-2h_1^2h_3\\1&0\\0&1\end{pmatrix}\]
This can be equivalently stated as $\jacobian{h}{}=q(h)$. This is a generalization of PIVP to polynomial
partial differential equations.

To complete this example, note that it can be solved exactly and $h_1(x_1,x_2)=\frac{1}{x_1^2+x_2^2}$
which is defined over $\R^2\setminus\{(0,0)\}$.
\end{example}

\begin{example}[GPAC composition]
Another way to look at \figref{fig:multidim_gpac_complex} and
\figref{fig:multidim_gpac_simple} is to imagine that $x_1=X_1(t)$ and $x_2=X_2(t)$
are functions of the time (produced by other GPACs), and rewrite the system in the time domain with $h=H(t)$:
\[
\begin{array}{rl}
H_2'(t)&=X_1'(t)\\
H_3'(t)&=X_2'(t)\\
H_1'(t)&=-2H_1(t)^2H_2(t)X_1'(t)-2H_1(t)^2H_3(t)X_2'(t)
\end{array}
\]
We obtain a system similar to the unidimensional PIVP: for a given
choice of $X$ we have $H'(t)=q(H(t))X'(t)$ where $q(h)$ is the polynomial matrix given by:
\[q(h)=\begin{pmatrix}-2h_1^2h_2&-2h_1^2h_3\\1&0\\0&1\end{pmatrix}\]
Note that this is the same polynomial matrix as in the previous example.
The relationship between the time domain $H$ and the original $h$
is simply given by $H(t)=h(x(t))$. This approach has a natural interpretation
on the GPAC circuit in terms of circuit rewriting. Assume that $x_1$ and $x_2$
are the outputs of two GPACs (with input $t$), i.e.~$x_1=x_1(t)$ and $x_2=x_2(t)$. Then
$x_1,x_2$ are given by the first two components of a polynomial ODE \eqref{eq:ode},
i.e.~$x_1(t)=y_1(t)$ and $x_2(t)=y_2(t)$. Moreover one has $x_1^{\prime}(t)=p_1(y),
x_2^{\prime}(t)=p_2(y)$. That means that the output $H(t)=(H_1(t),H_2(t),H_3(t))$ of
the GPAC of \figref{fig:multidim_gpac_complex} satisfies
\[
H'(t)=q(H(t))X'(t)=q(H(t))(p_1(y),p_2(y))
\]
and therefore consists of the first three components of the polynomial ODE given by
\[
\begin{array}{rl}
H^{\prime}&=q(H(t))(p_1(y),p_2(y))\\
y^{\prime}&=p(y)
\end{array}
\]
Thus, if $x_1$ and $x_2$ are the outputs of the some GPACs, depending on one input $t$, and if
we connect the outputs of these two GPACs to the inputs of the two-dimensional GPAC of
\figref{fig:multidim_gpac_complex}, we obtain a one-input GPAC computing $H(t)$, where $t$ is the input. Note that in
a normal GPAC,
the time $t$ is the only valid input of the derivative port of
the integrator, so we need to rewrite integrators which violate this rule. This can be done
by rewriting the ODE defining $H(t)$ into a polynomial ODE as done above, and then by implementing
a GPAC which computes the solution of this ODE such that the time $t$ is the only valid input of the
derivative port of each integrator (this is trivial to implement).
This procedure
always stops in finite time. Moreover it always works as long as $q(\cdot )$ is a matrix consisting of
polynomials.
\end{example}

These considerations lead to state that the following generalization is
clearly the one we want: 

\begin{definition}[Generable function]\label{def:gpac_generable_ext}
Let $d,\ell\in\N$, $I$ an open and connected subset of $\R^d$, $\mtt{sp}:\Rp\rightarrow\Rp$
a nondecreasing function and $f:I\rightarrow\R^\ell$. We say that $f\in\gval[\K]{\mtt{sp}}$ if and only if
there exists $n\geqslant \ell$, $p\in\MAT{n}{d}{\K}[\R^n]$, $x_0\in(\K^d\cap I)$, $y_0\in\K^n$
and $y:I\rightarrow\R^n$ satisfying for all $x\in I$:
\begin{itemize}
\item $y(x_0)=y_0$ and $\jacobian{y}(x)=p(y(x))$ (i.e. $\partial_jy_i(x)=p_{ij}(y(x))$)\hfill$\blacktriangleright$ $y$ satisfies a differential equation
\item $f(x)=y_{1..\ell}(x)$\hfill$\blacktriangleright$ $f$ is a component of $y$
\item $\infnorm{y(x)}\leqslant \mtt{sp}(\infnorm{x})$\hfill$\blacktriangleright$ $y$ is bounded by $\mtt{sp}$
\end{itemize}
\end{definition}

\begin{remark}[Uniqueness]
The uniqueness of $y$ in Definition \ref{def:gpac_generable_ext} can be seen in two different ways:
by uniqueness of the unidimensional case and by analyticity. Note that
the existence of $y$ (and thus the domain of definition) is a hypothesis of the definition.

Consider $x\in I$ and $\gamma$ a smooth curve\footnote{see \remref{rem:gpac_ext_dom}}
from $x_0$ to $x$ with values in $I$ and consider $z(t)=y(\gamma(t))$ for $t\in[0,1]$.
It can be seen that $z'(t)=\jacobian{y}(\gamma(t))\gamma'(t)=p(y(\gamma(t))\gamma'(t)=p(z(t))\gamma'(t)$,
$z(0)=y(x_0)=y_0$ and $z(1)=y(x)$.
The initial value problem $z(0)=y_0$ and $z'(t)=p(z(t))\gamma'(t)$ satisfies the hypothesis
of the Cauchy-Lipschitz theorem and as such admits a unique solution. Since this
IVP is independent of $y$, the value of $z(1)$ is unique and must be equal to $y(x)$, for any
solution $y$ and any $x$. This implies that $y$
 must be unique.
 
Alternatively, use Proposition \ref{prop:gpac_ext_analytic} to conclude that any solution
must be analytic. Assume that there are two solutions $y$ and $z$. Then all partial derivatives
at any order at the initial point $x_0$ are equal because they only depend on $y_0$.
Thus $y$ and $z$ have the same partial derivatives at all order and must be equal on
a small open ball around $y_0$. A classical argument of finite covering with open balls
then extends this argument to any point of the interior of domain of definition that is connected to $y_0$.
Since the domain of definition is assumed to be open and connected, this concludes to the equality of $y$ and $z$.
\end{remark}

\begin{remark}[Regularity]\label{rem:gpac_ext_regularity}
In the euclidean space $\R^n$, $C^k$ smoothness is equivalent to the smoothness
of the order $k$ partial derivatives. Consequently, the equation $\jacobian{y}=p(y)$
on the open set $I$ immediately proves that $y$ is $C^\infty$.
\propref{prop:gpac_ext_analytic} shows that $y$ is in fact real analytic.
\end{remark}

\begin{remark}[Domain of definition]\label{rem:gpac_ext_dom}
Definition \ref{def:gpac_generable_ext} requires the domain of definition of $f$ to be connected,
otherwise it would not make sense. Indeed, we can only define the value of $f$
at point $u$ if there exists a path from $x_0$ to $u$ in the domain of $f$.
It could seem, at first sight, that the domain being ``only''
connected may be too weak to work with. This is not the case, because
in the euclidean space $\R^d$, \emph{open} connected subsets are
always smoothly arc connected, that is any two points can be connected using a
smooth $C^1$ (and even $C^\infty$) arc. \propref{prop:connected_is_generable_connected}
extends this idea to generable arcs, with a very useful corollary.
\end{remark}

\begin{remark}[Multidimensional output]\label{rem:multidim_out_ext}
\remref{rem:multidim_out} also applies to this definition: $f:\subseteq\R^d\rightarrow\R^n$
is generable if and only if each of its component is generable
(\emph{i.e.} $f_i$ is generable for all $i$).
\end{remark}

\begin{remark}[Definition consistency]
It should be clear that Definition \ref{def:gpac_generable_ext} and Definition \ref{def:gpac_generable}
are consistent. More precisely, in the case of unidimensional function ($d=1$) with
domain of definition $I=\R$, both definitions are exactly the same since $\jacobian{y}=y'$
and $\MAT{n}{1}{\R}=\R^n$.
\end{remark}

The following example focuses on the second issue mentioned at the beginning of
the section, namely the domain of definition.

\begin{example}[Inverse and logarithm functions]\label{ex:inv_ln}
We illustrate that the choice of the domain of definition makes important differences
in the nature of the function.
\begin{itemize}
\item Let $0<\varepsilon<1$ and define $f_\varepsilon:x\in]\varepsilon,\infty[\mapsto\frac{1}{x}$. It
can be seen that $f_\varepsilon'(x)=-f_\varepsilon(x)^2$ and $f_\varepsilon(1)=1$.
Furthermore, $|f_\varepsilon(x)|\leqslant\frac{1}{\varepsilon}$ thus $f_\varepsilon\in\gval{\lambdafunex{\alpha}{\frac{1}\varepsilon}}$.
So in particular, $f_\varepsilon\in\gval{\poly}$ for any $\varepsilon>0$. Something interesting arises
when $\varepsilon\rightarrow0$: define $f_0(x)=x\in(0,\infty)\mapsto\frac{1}{x}$.
Then $f_0$ is still generable and $|f_0(x)|\leqslant\frac{1}{|x|}$. Thus $f_0\in\gval{\lambdafunex{\alpha}{\frac{1}{\alpha}}}$
but $f_0\notin\gval{\poly}$. Note that strictly speaking, $f_0\in\gval{\mtt{sp}}$ where $\mtt{sp}(\alpha)=\frac{1}{\alpha}$
and $\mtt{sp}(0)=0$ because the bound function needs to be defined over $\Rp$.
\item A similar phenomenon occurs with the logarithm: define $g_\varepsilon:x\in(\varepsilon,\infty)\mapsto\ln(x)$.
Then $g_\varepsilon'(x)=f_\varepsilon(x)$ and $g_\varepsilon(1)=0$. Furthermore, $|g_\varepsilon(x)|\leqslant\max(|x|,|\ln\varepsilon|)$.
Thus $g_\varepsilon\in\gval{\lambdafunex{\alpha}{\max(\alpha,|\ln\varepsilon|,\frac{1}{\varepsilon})}}$, and in
particular $g_\varepsilon\in\gval{\poly}$ for any $\varepsilon>0$. Similarly, $g_0:x\in]0,\infty[\mapsto\ln(x)$
is generable but does not belong to $\gval{\poly}$.
\end{itemize}
\end{example}

\begin{example}[Classical non-generable functions]
While many of the usual real functions are known to be generated by a
GPAC, a notable exception is Euler's Gamma function $\Gamma(x)=\int_{0}^{\infty}t^{x-1}e^{-t}dt$
function or Riemann's Zeta function $\zeta(x)=\sum_{k=0}^\infty \frac1{k^x}$ \citep{Sha41}, \citep{PR89}.
Furthermore, Riemann's Zeta function (over, for example, $[2,\infty)$) is an example
of real-analytic, polynomially-bounded that is not in \gval{\poly}.

\end{example}

\begin{example}[Generable functions not in \gval{\poly}]\label{ex-notin-gpval}

We have seen that Riemann's Zeta function $\zeta$ is an example of a function not in \gval{\poly} due to the
fact that it is not generable. An example of a generable function not belonging to \gval{\poly} is the exponential
$e^x$ because, while it is generable, its derivative is not bounded by another polynomial.
Note that it is quite possible to have bounded generable functions which do not belong to \gval{\poly}.
An example is the function given by $f(x)=\sin(e^x)$ which is generable and bounded, but its derivative
$f^{\prime}(x)=e^x\cos(e^x)$ is not bounded by any polynomial.
\end{example}

The previous examples show that $\gval[\K]{\mtt{sp}}$ can be used to define a proper hierarchy of generable
functions. Adapting the examples given in Example \ref{ex-notin-gpval} one can show for instance that

\[
\gval{\poly} \subsetneqq \gval{e^x} \subsetneqq \gval{e^{e^x}}  \subsetneqq \ldots
\]

In particular these examples show the following result.

\begin{theorem}[Existence of noncollapsing classes]
  $\gval{\poly} \subsetneqq \gval{}$.
\end{theorem}

\section{Stability properties}\label{sec:pivp_stable}

In this section, the major results will the be stability of multidimensional
generable functions under arithmetical operators, composition and ODE solving.
Note that some of the results use properties on $\K$ which can be found in
Section \ref{sec:gen_field_ext_stab}.

\begin{lemma}[Arithmetic on generable functions]\label{lem:gpac_ext_class_stable}
Let $d,\ell,n,m\in\N$, $\mtt{sp},\ovl{\mtt{sp}}:\R\rightarrow\Rp$,
$f:\subseteq\R^d\rightarrow\R^n\in\gval{\mtt{sp}}$ and
$g:\subseteq\R^\ell\rightarrow\R^m\in\gval{\ovl{\mtt{sp}}}$. Then:
\begin{itemize}
\item $f+g, f-g\in\gval{\mtt{sp}+\ovl{\mtt{sp}}}$ over $\dom{f}\cap\dom{g}$ if $d=\ell$ and $n=m$
\item $fg\in\gval{\max(\mtt{sp},\ovl{\mtt{sp}},\mtt{sp}\thinspace\ovl{\mtt{sp}})}$ if $d=\ell$ and $n=m$
\item $f\circ g\in\gval{\max(\ovl{\mtt{sp}},\mtt{sp}\circ\ovl{\mtt{sp}})}$ if $m=d$ and $g(\dom{g})\subseteq \dom{f}$
\end{itemize}
\end{lemma}

\begin{proof}
We focus on the case of the composition, the other cases are very similar.

Apply Definition \ref{def:gpac_generable_ext} to $f$ and $g$ to respectively get
$l,\bar{l}\in\N$, $p\in\MAT{l}{d}{\K}[\R^l]$, $\bar{p}\in\MAT{\bar{l}}{\ell}{\K}[\R^{\bar{l}}]$,
$x_0\in\dom{f}\cap\K^d$, $\bar{x}_0\in\dom{g}\cap\K^\ell$, $y_0\in\K^l$, $\bar{y}_0\in\K^{\bar{l}}$,
$y:\dom{f}\rightarrow\R^l$ and $\bar{y}:\dom{g}\rightarrow\R^{\bar{l}}$.
Define $h=y\circ g$, then $\jacobian{h}=\jacobian{y}(g)\jacobian{g}=p(h)\bar{p}_{1..m}(\bar{y})$
and $h(\bar{x_0})=y(\bar{y}_0)\in\K^l$ by \amaurycorref{cor:generable_field_ext_stab}.
In other words $(\bar{y},h)$ satisfy:
\[\left\{\begin{array}{@{}r@{}l}\bar{y}(\bar{x}_0)&=y_0\in\K^{\bar{l}}\\h(\bar{x}_0)&=y(\bar{y}_0)\in\K^l\end{array}\right.
\qquad\left\{\begin{array}{@{}r@{}l}\bar{y}'&=\bar{p}(\bar{y})\\h'&=p(h)\bar{p}_{1..m}(\bar{y})\end{array}\right.\]
This shows that $f\circ g=z_{1..m}\in\gval{}$. Furthermore,
\begin{align*}
\infnorm{(\bar{y}(x),h(x))}
    &\leqslant\max(\infnorm{\bar{y}(x)},\infnorm{y(g(x))})\\
    &\leqslant\max(\mtt{\ovl{sp}}(\infnorm{x}),\mtt{sp}(\infnorm{g(x)}))\\
    &\leqslant\max(\mtt{\ovl{sp}}(\infnorm{x}),\mtt{sp}(\mtt{\ovl{sp}}(\infnorm{x}))).
\end{align*}
\end{proof}

Our main result is that the solution to an ODE whose right hand-side is
generable, and possibly depends on an external and $C^1$ control, may be rewritten
as a GPAC. A corollary of this result is that the solution to a generable ODE is
generable.

\begin{proposition}[Generable ODE rewriting]\label{prop:gpac_ext_ivp_stable_pre}
Let $d,n\in\N$, $I\subseteq\R^n$, $X\subseteq\R^d$, $\mtt{sp}:\Rp\rightarrow\Rp$ and
$(f:I\times X\rightarrow\R^n)\in\gval[\K]{\mtt{sp}}$. Define $\ovl{\mtt{sp}}=\max(\idfun,\mtt{sp})$.
Then there exists $m\in\N$, $(g:I\times X\rightarrow\R^m)\in\gval[\K]{\ovl{\mtt{sp}}}$
and $p\in\K^m[\R^m\times\R^d]$ such that for any interval $J$,
$t_0\in\K\cap J$, $y_0\in\K^n\cap J$, $y\in C^1(J,I)$ and $x\in C^1(J,X)$,
if $y$ satisfies:
\[\left\{\begin{array}{@{}r@{}l@{}}y(t_0)&=y_0\\y'(t)&=f(y(t),x(t))\end{array}\right.
\qquad\forall t\in J\]
then there exists $z\in C^1(J,\R^m)$ such that:
\[\left\{\begin{array}{@{}r@{}l@{}}z(t_0)&=g(y_0,x(t_0))
\\z'(t)&=p(z(t),x'(t))\end{array}\right.
\qquad \left\{\begin{array}{@{}r@{}l@{}}y(t)&=z_{1..d}(t)\\\infnorm{z(t)}&\leqslant\ovl{\mtt{sp}}(\max(\infnorm{y(t)},\infnorm{x(t)}))\end{array}\right.
\qquad \forall t\in J\]
\end{proposition}

\begin{proof}
Apply Definition \ref{def:gpac_generable_ext} to $f$ get $m\in\N$, $p\in\MAT{m}{n+d}{\K}[\R^m]$,
$f_0\in\dom{f}\cap\K^d$, $w_0\in\K^m$ and $w:\dom{f}\rightarrow\R^m$ such that $w(f_0)=w_0$,
$\jacobian{w(v)}=p(w(v))$, $\infnorm{w(v)}\leqslant\mtt{sp}(\infnorm{v})$ and $w_{1..n}(v)=f(v)$ for all $v\in\dom{f}$.
Define $u(t)=w(y(t),x(t))$, then:
\begin{align*}
u'(t)&=\jacobian{w}(y(t),x(t))(y'(t),x'(t))\\
&=p(w(y(t),x(t)))(f(y(t),x(t)),x'(t))\\
&=p(u(t))(u_{1..n}(t),x'(t))\\
&=q(u(t),x'(t))
\end{align*}
where $q\in\K^m[\R^{m+d}]$ and $u(t_0)=w(y(t_0))=w(y_0,x(t_0))$. Note that $w$ itself
is a generable function and more precisely $w\in\gval[\K]{\poly}$ by definition.
Finally, note that $y'(t)=u_{1..d}(t)$ so that we get for all $t\in J$:
\[\left\{\begin{array}{@{}r@{}l@{}}y(t_0)&=y_0\\
y'(t)&=u_{1..d}(t)\end{array}\right.
\qquad \left\{\begin{array}{@{}r@{}l@{}}u(t_0)&=w(y_0,x(t_0))\\
u'(t)&=q(u(t),x'(t))\end{array}\right.\]
Define $z(t)=(y(t),u(t))$, then $z(t_0)=(y_0,w(y_0,x(t_0)))=g(y_0,x(t_0))$
where  $y_0\in\K^n$ and $w\in\gval[\K]{\mtt{sp}}$ so $g\in\gval[\K]{\ovl{\mtt{sp}}}$.
And clearly $z'(t)=r(z(t),x'(t))$ where $r\in\K^{n+m}[\R^{n+m}]$.
Finally, $\infnorm{z(t)}=\max(\infnorm{y(t)},\infnorm{w(y(t),x(t))})
\leqslant\max(\infnorm{y(t)},\mtt{sp}(\max(\infnorm{y(t)},\infnorm{x(t)}))
\leqslant\ovl{\mtt{sp}}(\max(\infnorm{y(t)},\infnorm{x(t)}))$.
\end{proof}

A simplified version of this lemma shows that generable functions are closed under
ODE solving.

\begin{corollary}[Generable functions are closed under ODE]\label{cor:gpac_ext_ivp_stable}
Let $d\in\N$, $J\subseteq\R$ an interval, $\mtt{sp},\ovl{\mtt{sp}}:\Rp\rightarrow\Rp$,
$f:\subseteq\R^d\rightarrow\R^d$ in \gval{\mtt{sp}}, $t_0\in\K\cap J$ and $y_0\in\K^d\cap\dom{f}$.
Assume there exists $y:J\rightarrow\dom{f}$ satisfying for all $t\in J$:
\[\left\{\begin{array}{@{}r@{}l@{}}y(t_0)&=y_0\\y'(t)&=f(y(t))\end{array}\right.
\qquad\infnorm{y(t)}\leqslant\ovl{\mtt{sp}}(t)\]
Then $y\in\gval{\max(\ovl{\mtt{sp}},\mtt{sp}\circ \ovl{\mtt{sp}})}$ and is unique.
\end{corollary}

\begin{remark}[Polynomially bounded generable functions] \label{rq:vingtquatre}
In light of the stability properties above, the class of \emph{polynomially bounded}
generable functions, $$\gval{\poly}=\bigcup_{k=1}^\infty\gval{\lambdafunex{\alpha}{k\alpha^k}}$$
is particularly interesting because it is stable by operations: addition, multiplication,
composition and ODE solving (provided the solution is polynomially bounded). Notice
that $\gval{\poly}$ is not simply the intersection of $\gval{}$ with the set of
functions bounded by a polynomial, as shown in Example \ref{ex-notin-gpval}.
\end{remark}

Our last result is simple but very useful. Generable functions are continuous
and continuously differentiable, so locally Lipschitz continuous. We can give a
precise expression for the modulus of continuity in the case where the domain of
definition is simple enough.


\begin{proposition}[Modulus of continuity]\label{prop:generable_mod_cont}
Let $\mtt{sp}:\Rp\rightarrow\Rp$, $f\in\gval{\mtt{sp}}$. There exists $q\in\K[\R]$ such that
for any $x_1,x_2\in\dom{f}$, if $[x_1,x_2]\subseteq\dom{f}$ then
    $\infnorm{f(x_1)-f(x_2)}\leqslant\infnorm{x_1-x_2}q(\mtt{sp}(\max(\infnorm{x_1},\infnorm{x_2})))$.
In particular, if $\dom{f}$ is convex then $f$ has a polynomial modulus of continuity.
\end{proposition}

\begin{proof}
Apply Definition \ref{def:gpac_generable_ext} to get $d,\ell,n,p,x_0,y_0$ and $y$. Let $k=\degp{p}$. Recall that
for a matrix, the subordinate norm is given by $|||M|||=\max_i\sum_j|M_{ij}|$. Then:
\begin{align*}
\infnorm{f(x_1)-f(x_2)}&=\infnorm{\int_{x_1}^{x_2}\jacobian{y_{1..\ell}}(x)dx}
    =\infnorm{\int_0^1\jacobian{y_{1..\ell}}((1-\alpha)x_1+\alpha x_2)(x_2-x_1)d\alpha}\\
    &\leqslant\int_0^1|||\jacobian{y_{1..\ell}}((1-\alpha)x_1+\alpha x_2)|||\cdot\infnorm{x_2-x_1}d\alpha\\
    &\leqslant\infnorm{x_2-x_1}\int_0^1\max_{i\in\intinterv{1}{\ell}}\sum_{j=1}^d|p_{ij}(y((1-\alpha)x_1+\alpha x_2))|d\alpha\\
    &\leqslant\infnorm{x_2-x_1}\int_0^1\max_{i\in\intinterv{1}{\ell}}\sum_{j=1}^d\sigmap{p}\max(1,\infnorm{y((1-\alpha)x_1+\alpha x_2)})^k)d\alpha\\
    &\leqslant\infnorm{x_2-x_1}\int_0^1\max_{i\in\intinterv{1}{\ell}}d\sigmap{p}\max(1,\mtt{sp}(\infnorm{(1-\alpha)x_1+\alpha x_2}))^kd\alpha\\
    &\leqslant\infnorm{x_2-x_1}\int_0^1d\sigmap{p}\max(1,\mtt{sp}(\max(\infnorm{x_1},\infnorm{x_2})))^kd\alpha\\
    &\leqslant\infnorm{x_2-x_1}d\sigmap{p}\max(1,\mtt{sp}(\max(\infnorm{x_1},\infnorm{x_2})))^k\\
\end{align*}
\end{proof}

\section{Analyticity of generable functions}\label{sec:pivp_analytic}

It is a well-known result that the solution of a PIVP $y'=p(y)$ (and more generally,
of an analytic differential equation $y'=f(y)$ where $f$ is analytic) is real analytic on
its domain of definition. In the previous section we defined a generalized notion
of generable function satisfying $\jacobian{y}=p(y)$ which analyticity is less immediate.
In this section we go through the proof in detail, which of course subsumes
the result for PIVP.

We recall a well-known characterization of analytic functions. It is indeed much
easier to show that a function is infinitely differentiable and of controlled growth,
rather than showing the convergence of the Taylor series.

\begin{proposition}[Characterization of analytic functions]\label{prop:analytic_characterization}
Let $f\in C^\infty(U)$ for some open subset $U$ of $\R^m$. Then $f$ is analytic on $U$
if and only if, for each $u\in U$, there are an open ball $V$, with $u\in V\subseteq U$,
and constants $C>0$ and $R>0$ such that the derivatives of $f$ satisfy
\[|\partial_\alpha f(x)|\leqslant C\frac{\alpha!}{R^{|\alpha|}}\qquad x\in V,\alpha\in\N^m\]
\end{proposition}

\begin{proof}See proposition 2.2.10 of \cite{KP02}.
\end{proof}

In order to use this result, we show that the derivatives of generable functions
at a point $x$ do not grow faster than the described bound. We
use a generalization of Faà di Bruno formula for the derivatives of a composition.

\begin{theorem}[Generalised Faà di Bruno's formula]\label{th:generalised_faa_di_bruno}
Let $f:X\subseteq\R^d\rightarrow Y\subseteq\R^n$ and $g:Y\rightarrow\R$ where $X,Y$
are open sets and $f,g$ are sufficiently smooth functions\footnote{More precisely,
for the formula to hold for $\alpha$, all the derivatives which appear in the right-hand
side must exist and be continuous}. Let $\alpha\in\N^d$ and $x\in X$, then
\[\partial_\alpha(g\circ f)(x)=\alpha!\smashoperator{\sum_{(s,\beta,\lambda)\in\mathcal{D}_\alpha}}
    \partial_\lambda g(f(x))\prod_{k=1}^s\frac{1}{\lambda_{k}!}
    \left(\frac{1}{\beta_k!}\partial_{\beta_k}f(x)\right)^{\lambda_{k}}\]
where $\partial_\lambda$ means $\partial_{\sum_{u=1}^s\lambda_u}$
and where $\mathcal{D}_\alpha$ is the list of \emph{decompositions} of $\alpha$.
A multi-index $\alpha\in\N^d$ is decomposed into $s\in\N$ parts $\beta_1,\ldots,\beta_s\in\N^d$
with multiplicies $\lambda_1,\ldots,\lambda_s\in\N^n$ respectively if $|\lambda_i|>0$
for all $i$, all the $\beta_i$ are distincts from each other and from $0$, and
$\alpha=|\lambda_1|\beta_1+\cdots+|\lambda_s|\beta_s$. Note that $\beta$ and $\lambda$
are multi-indices of multi-indices: $\beta\in\left(\N^d\right)^s$ and $\lambda\in\left(\N^d\right)^s$.
\end{theorem}

\begin{proof}
See \cite{Ma09} or \cite{Encinas2003975}.
\end{proof}

We have seen that one-dimensional GPAC generable functions are analytic. We now extend this
result to the multidimensional case.

\begin{proposition}[Generable implies analytic]\label{prop:gpac_ext_analytic}
If $f\in\gval{}$ then $f$ is real-analytic on $\dom{f}$.
\end{proposition}

\begin{proof}
Let $\mtt{sp}:\R\rightarrow\Rp$, $p\in\MAT{n}{d}{}[\R^n]$ and $y:\R^n\rightarrow\R^n$
from Definition \ref{def:gpac_generable_ext}. It is sufficient to
prove that $y$ is analytic on $D=\dom{f}$ to get the result. Let $i\in\intinterv{1}{n}$,
and $j\in\intinterv{1}{d}$, since $\jacobian{y}=p(y)$ then $\partial_jy_i(x)=p_{ij}(y(x))$
and $p_{ij}$ is a polynomial vector so clearly $C^\infty$.
By \remref{rem:gpac_ext_regularity}, $y$ is also $C^\infty$ so we can apply \thref{th:generalised_faa_di_bruno}
for any $x\in D$, $\alpha\in\N^d$ and get
\[\partial_\alpha(\partial_jy_i)(x)=\partial_\alpha (p_{ij}\circ y)(x)=\alpha!\sum_{(s,\beta,\lambda)\in\mathcal{D}_\alpha}
    \partial_\lambda p_{ij}(y(x))\prod_{k=1}^s\frac{1}{\lambda_{k}!}
    \left(\frac{1}{\beta_k!}\partial_{\beta_k}y(x)\right)^{\lambda_{k}}\]
Define $B_\alpha(x)=\frac{1}{\alpha!}\infnorm{\partial_\alpha y(x)}$, and denote by $\alpha+j$ the
multi-index $\lambda$ such that $\lambda_j=\alpha_j+1$ and $\lambda_k=\alpha_k$ for $k\neq j$.
Define $C(y(x))=\max_{i,j,\lambda}(|\partial_\lambda p_{ij}(y(x))|)$ and note that it is well-defined
because $\partial_\lambda p_{ij}$ is zero whenever $|\lambda|>\degp{p_{ij}}$.
Define $\mathcal{D}_\alpha'=\{(s,\beta,\lambda)\in\mathcal{D}_\alpha\thinspace|\thinspace|\lambda|\leqslant\degp{p}\}$.
The equations becomes:
\begin{align*}
|\partial_\alpha(\partial_jy_i)(x)|
    &\leqslant\alpha!\sum_{(s,\beta,\lambda)\in\mathcal{D}_\alpha}
    |\partial_\lambda p_{ij}(y(x))|\prod_{k=1}^s\frac{1}{\lambda_{k}!}
    \left|\frac{1}{\beta_k!}\partial_{\beta_k}y(x)\right|^{\lambda_{k}}\\
    &\leqslant\alpha!C(y(x))\sum_{(s,\beta,\lambda)\in\mathcal{D}'_\alpha}
    \prod_{k=1}^s\frac{1}{\lambda_{k}!}
    B_{\beta_k}(x)^{|\lambda_{k}|}.\\
\end{align*}
Note that the right-hand side of the expression does not depend on $i$. We are going
to show by induction that $B_\alpha(x)\leqslant\left(\frac{C(y(x))}{R}\right)^{|\alpha|}$
for some choice of $R$.
The initialization for $|\alpha|=1$ is trivial because $\alpha!=1$ and
$B_\alpha(x)=\infnorm{\partial_\alpha y(x)}\leqslant C(y(x))$ so we only need $R\leqslant1$.
The induction step is as follows:
\begin{align*}
B_{\alpha+j}(x)&\leqslant C(y(x))\sum_{(s,\beta,\lambda)\in\mathcal{D}'_\alpha}
    \prod_{k=1}^s\frac{1}{\lambda_{k}!}
    B_{\beta_k}(x)^{|\lambda_{k}|}\\
    &\leqslant C(y(x))\sum_{(s,\beta,\lambda)\in\mathcal{D}'_\alpha}
    \prod_{k=1}^s\frac{1}{\lambda_{k}!}\left(\frac{C(y(x))}{R}\right)^{|\beta_k||\lambda_k|}\\
    &\leqslant C(y(x))\sum_{(s,\beta,\lambda)\in\mathcal{D}'_\alpha}
    \frac{1}{\lambda!}\left(\frac{C(y(x))}{R}\right)^{\sum_{u=1}^s|\beta_k||\lambda_k|}\\
    &\leqslant C(y(x))\left(\frac{C(y(x))}{R}\right)^{|\alpha|}\sum_{(s,\beta,\lambda)\in\mathcal{D}'_\alpha}
    \frac{1}{\lambda!}\\
    &\leqslant C(y(x))\left(\frac{C(y(x))}{R}\right)^{|\alpha|}\card{\mathcal{D}'_\alpha}.
\end{align*}

Evaluating the exact cardinal of $\mathcal{D}'_\alpha$ is complicated but we only
need a good enough bound to get on with it. First notice that for any
$(s,\beta,\lambda)\in\mathcal{D}'_\alpha$, we have $|\lambda|\leqslant\degp{p}$ by
definition, and since each $|\lambda_i|>0$, necessarily $s\leqslant\degp{p}$.
This means that there is a finite number, denote it by $A$, of $(s,\lambda)$ in $\mathcal{D}'_\alpha$.
For a given $\lambda$, we must have $\alpha=\sum_{i=1}^s|\lambda_i|\beta_i$ which implies that
$|\beta_{ij}|\leqslant|\alpha|$ and so there at most $(1+|\alpha|)^{ns}$ choices for $\beta$,
and since $s\leqslant\degp{p}$, $\card{\mathcal{D}'_\alpha}\leqslant A(1+|\alpha|)^b$
where $b$ and $A$ are constants. Choose $R\leqslant1$ such that $R^{|\alpha|}\geqslant A(1+|\alpha|)^b$ for all $\alpha$
to get the claimed bound on $B_\alpha(x)$.

To conclude with \propref{prop:analytic_characterization}, consider $x\in D$.
Let $V$ be an open ball of $D$ containing $x$. Let $M=\sup_{u\in V}C(y(x))$,
it is finite because $C$ is bounded by a polynomial, $\infnorm{y(x)}\leqslant\mtt{sp}(x)$
and $V$ is an open ball (thus included in a compact set). Finally we get:
\[\infnorm{\partial_\alpha y(x)}\leqslant\alpha!\left(\frac{M}{R}\right)^{|\alpha|}\]

\end{proof}

\section{Generable zoo}\label{sec:helper_func}

In this section, we introduce a number of generable functions. Since a GPAC (PIVP) only
generates analytic functions, it cannot generate discontinuous functions like the sign. However
these functions can be arbitrarily approximated by GPACs, as we show in this section, where
we present a ``zoo'' of such approximating functions. This zoo illustrates the wide range of generable functions.
Some of the functions selected in this ``zoo'' were chosen to approximate noncontinuous functions traditionally used
in computer programs like the absolute value or the sign function. Other functions were selected
due to their usefulness for potential applications, like simulating Turing machines with a GPAC,
using a bounded amount of resources, which we intend to explore in an incoming paper.

We note that the approximation of a discontinuous functions by a GPAC generable function is uniform, since we
provide the GPAC with a parameter which sets the maximum allowed error of the approximation. The use of different
values of the parameter by the same GPAC allows to dynamically change the quality of the approximation, without
making any other change on the GPAC. The table below gives a list of the functions and their purpose.

We use the term ``dead zone'' to refer to interval(s) where the generable
function does not compute the expected function (but still has controlled behavior).
We use the term ``high'' to mean that the function is close to $x$ (an input)
within $e^{-\mu}$ where $\mu$ is another input. Conversely, the use the term ``low''
to mean that it is close to $0$ within $e^{-\mu}$. And ``X'' means something in between.
Finally ``integral'' means that function is of the form $\phi x$ and
the integral of $\phi$ (on some interval) is between $1$ and a constant.

We conclude this section by giving a large class of functions that can be
uniformly approximated by (polynomially bounded) generable functions, except on a small number of dead zones
(typically at discontinuity points) that can be made arbitrary small, see Section \ref{sec:zoo_generic}.

\begin{longtable}{@{}p{2.5cm}p{2.5cm}p{7cm}@{}}
\multicolumn{3}{@{}l}{\textbf{\Large Generable Zoo}}\\
Name & Notation & Comment \\
\midrule
\endhead
Sign & $\sg(x,\mu,\lambda)$ & Compute the sign of $x$ with error $e^{-\mu}$
    and dead zone in $[-\lambda^{-1},\lambda^{-1}]$. See \ref{lem:sg} \\
Floor & $\ip{1}(x,\mu,\lambda)$ & Compute $\intp_1(x)$ with error $e^{-\mu}$
    and dead zone in $[-\lambda^{-1},\lambda^{-1}]$. See \ref{cor:ip1}\\
Abs & $\abs(x,\mu,\lambda)$ & Compute $|x|$ with error with error $e^{-\mu}$
    and dead zone in $[-\lambda^{-1},\lambda^{-1}]$. See \ref{lem:abs}\\
Max & $\mx(x,y,\mu,\lambda)$ & Compute $\max(x,y)$ and  $\infnorm{x}$ with error $e^{-\mu}$
and dead zone for $x-y\in[-\lambda^{-1},\lambda^{-1}]$. See \ref{lem:max}\\
Norm & $\norm_\delta(x,\mu,\lambda)$ & Compute $\infnorm{x}$ with error $\delta$. See \ref{lem:norm}\\
Round & $\rnd(x,\mu,\lambda)$ & Compute $\round{x}$ with error $e^{-\mu}$ and dead zones in
    $[n-\frac{1}{2}+\lambda^{-1},n+\frac{1}{2}-\lambda^{-1}]$ for all $n\in\Z$. See \ref{lem:rnd}\\
Low-X-High & $\lxh_{[a,b]}(t,\mu,x)$ & Compute $0$ when $t\in]-\infty,a]$ and $x$ when $t\in[b,\infty[$ with error $e^{-\mu}$
    and a dead zone in $[a,b]$. See \ref{lem:lxh_hxl}\\
High-X-Low & $\hxl_{[a,b]}(t,\mu,x)$ & Compute $x$ when $t\in]-\infty,a]$ and $0$ when $t\in[b,\infty[$ with error $e^{-\mu}$
    and a dead zone in $[a,b]$. See \ref{lem:lxh_hxl}\\
\bottomrule
\end{longtable}

\subsection{Sign and rounding}

We begin with a small result on the hyperbolic tangent function, which will be
used to build several generable functions of interest.

\begin{lemma}[Bounds on $\tanh$]\label{lem:bounds_tanh}
$1-\sgn{t}\tanh(t)\leqslant e^{-|t|}$ for all $t\in\R$.
\end{lemma}

\begin{proof}
The case of $t=0$ is trivial. Assume that $t\geqslant 0$ and observe that
$1-\tanh(t)=1-\frac{1-e^{-2t}}{1+e^{-2t}}=\frac{2e^{-2t}}{1+e^{-2t}}=e^{-t}\frac{2e^{-t}}{1+e^{-2t}}$.
Define $f(t)=\frac{2e^{-t}}{1+e^{-2t}}$ and check that $f'(t)=\frac{2e^{-t}(e^{-2t}-1)}{(1+e^{-2t})^2}\leqslant0$
for $t\geqslant0$. Thus $f$ is a non-increasing function and $f(0)=1$ which concludes.

If $t<0$ then note that $1-\sgn{t}\tanh(t)=1-\sgn{-t}\tanh(-t)$
so we can apply the result to $-t\geqslant0$ to conclude.
\end{proof}

The simplest generable function of interest uses the hyperbolic tangent
to approximate the sign function. On top of the sign function, we can build
an approximation of the floor function. See \figref{fig:sg} for a graphical representation.

\begin{definition}[Sign function]\label{def:sg}
For any $x,\mu,\lambda\in\R$ define
\[\sg(x,\mu,\lambda)=\tanh(x\mu\lambda)\]
\end{definition}

\begin{lemma}[Sign]\label{lem:sg}
$\sg\in\gval{\poly}$ and for any $x\in\R$ and $\lambda,\mu\geqslant0$,
\[|\sgn{x}-\sg(x,\mu,\lambda)|\leqslant e^{-|x|\lambda\mu}\leqslant1\]
In particular, $\sg$ is non-decreasing in $x$ and if $|x|\geqslant\lambda^{-1}$ then
\[|\sgn{x}-\sg(x,\mu,\lambda)|\leqslant e^{-\mu}\]
\end{lemma}

\begin{proof}
Note that $\sg=\tanh\circ f$ where $f(x,\mu,\lambda)=x\mu\lambda$. We saw in Example \ref{ex:elem_func_gpac}
that $\tanh\in\gval{\lambdafunex{t}{1}}$. By \lemref{lem:gpac_ext_class_stable},
$f\in\gval{\alpha\mapsto\max(1,\alpha^3)}$. Thus $\sg\in\gval{\lambdafunex{\alpha}{\max(1,\alpha^3)}}$.

Use \lemref{lem:bounds_tanh} and the fact that $\tanh$ is an odd function to get the
first bound. The second bound derives easily from the first. Finally, $\sg$ is
a non-decreasing function because $\tanh$ is an increasing function.
\end{proof}

\begin{definition}[Floor function]\label{def:ip1}
For any $x,\mu,\lambda\in\R$ define
\[\ip{1}(x,\mu,\lambda)=\frac{1+\sg(x-1,\mu,\lambda)}{2}\]
\end{definition}

\begin{lemma}[Floor]\label{cor:ip1}
$\ip{1}\in\gval{\poly}$ and for any $x\in\R$ and $\mu,\lambda\geqslant0$,
\[|\intp_1(x)-\ip{1}(x,\mu,\lambda)|\leqslant\frac{e^{-|x-1|\lambda\mu}}{2}\leqslant\frac{1}{2}\]
where $\intp_1(x)=0$ if $x<1$ and $1$ if $x\geqslant1$.
In particular $\ip{1}$ is non-decreasing in $x$ and if $|1-x|\geqslant\lambda^{-1}$ then
\[|\intp_1(x)-\ip{1}(x,\mu,\lambda)|<e^{-\mu}\]
\end{lemma}

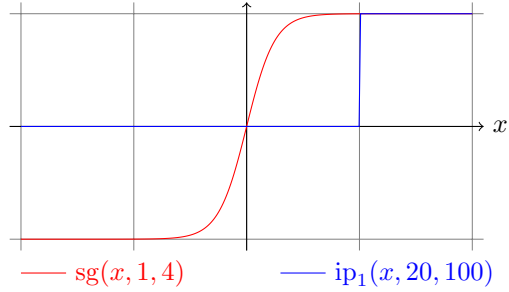
\begin{figure}
\begin{center}
\begin{tikzpicture}[domain=-2:2,samples=300,scale=1.5]
\draw[very thin,color=gray] (-2.1,-1.1) grid (2.1,1.1);
\draw[->] (0,-1.1) -- (0,1.1);
\draw[->] (-2.1,0) -- (2.1,0) node[right] {$x$};
\draw[color=red] plot[id=fn_xi_1] function{\fnsg{x}{1}{4}};
\draw[color=blue] plot[id=fn_sigma1_1] function{\fnipone{x}{20}{100}};
\draw[color=red] (-2,-1.3) -- (-1.6,-1.3) node[right] {$\sg(x,1,4)$};
\draw[color=blue] (0.3,-1.3) -- (0.7,-1.3) node[right] {$\ip{1}(x,20,100)$};
\end{tikzpicture}
\end{center}
\caption{Graph of $\sg$ and $\ip{1}$.}
\label{fig:sg}
\end{figure}

We will now see how to build a very precise approximation of the rounding function.
Of course rounding is not a continuous operation so we need a small deadzone around the discontinuity points.

\begin{definition}[Round function]\label{def:rnd}
For any $x\in\R$, $\lambda\geqslant2$ and $\mu\geqslant0$, define
\[\rnd(x,\mu,\lambda)=x-\frac{1}{\pi}\arctan(\operatorname{cltan}(\pi x,\mu,\lambda))\]
\[\operatorname{cltan}(\theta,\mu,\lambda)=
    \frac{\sin(\theta)}{\sqrt{\operatorname{nz}(\cos^2\theta,\mu+16\lambda^3,4\lambda^2)}}\sg(\cos\theta,\mu+3\lambda,2\lambda)\]
\[\operatorname{nz}(x,\mu,\lambda)=x+\frac{2}{\lambda}\ip{1}\left(1-x+\frac{3}{4\lambda},\mu+1,4\lambda\right)\]
\end{definition}

\begin{lemma}[Round]\label{lem:rnd}
For any $n\in\Z$, $\lambda\geqslant2$, $\mu\geqslant0$,
$|\rnd(x,\mu,\lambda)-n|\leqslant\frac{1}{2}$ for all $x\in\left[n-\frac{1}{2},n+\frac{1}{2}\right]$
and $|\rnd(x,\mu,\lambda)-n|\leqslant e^{-\mu}$ for all $x\in\left[n-\frac{1}{2}+\frac{1}{\lambda},n+\frac{1}{2}-\frac{1}{\lambda}\right]$.
Furthermore $\rnd\in\gval{\poly}$.
\end{lemma}

\begin{proof}
Let's start with the intuition first: consider $f(x)=x-\frac{1}{\pi}\arctan(\tan(\pi x))$.
It is an exact rounding function: if $x=n+\delta$ with $n\in\N$ and $\delta\in]\frac{-1}{2},\frac{1}{2}[$
then $\tan(\pi x)=\tan(\pi\delta)$ and since $\delta\pi\in]\frac{-\pi}{2},\frac{\pi}{2}[$, $f(x)=x-\delta=n$.
The problem is that it is undefined on all points of the form $n+\frac{1}{2}$ because
of the tangent function.

The idea is to replace $\tan(\pi x)$ by some ``clamped'' tangent $\operatorname{cltan}$
which will be like $\tan(\pi x)$ around integer points and stay bounded when close to $x=n+\frac{1}{2}$ instead of exploding.
To do so, we use the fact that $\tan\theta=\frac{\sin\theta}{\cos\theta}$ but this formula is problematic
because we cannot prevent the cosine from being zero, without loosing the sign of the expression (the cosine could never change sign).
Thus the idea is to remove the sign from the cosine, and restore it, so that $\tan\theta=\sgn{\cos\theta}\frac{\sin\theta}{|\cos\theta|}$.
And now we can replace $|\cos(\theta)|$ by $\sqrt{\operatorname{nz}(\cos^2\theta)}$, where $\operatorname{nz}(x)$
is mostly $x$ except near $0$ where is lower-bounded by some small constant (so it is never zero).
The sign of cosine can be computed using our approximate sign function $\sg$.

Formally, we begin with $\operatorname{nz}$ and show that:
\begin{itemize}
\item $\operatorname{nz}\in\gval{\poly}$
\item $\operatorname{nz}$ is an increasing function of $x$
\item For $x\geqslant\frac{1}{\lambda}$, $|\operatorname{nz}(x,\mu,\lambda)-x|\leqslant e^{-\mu}$
\item For $x\geqslant0$, $\operatorname{nz}(x,\mu,\lambda)\geqslant\frac{1}{2\lambda}$
\end{itemize}
The first point is a consequence of $\ip{1}\in\gval{\poly}$ from \amaurycorref{cor:ip1}. 
The second point comes from \amaurycorref{cor:ip1}:
if $x\geqslant\frac{1}{\lambda}$, then $1-x+\frac{3}{4\lambda}\leqslant 1-\frac{1}{4\lambda}$,
thus $|\operatorname{nz}(x,\mu,\lambda)-x|\leqslant \frac{2}{\lambda}e^{-\mu-1}\leqslant e^{-\mu}$
since $\lambda\geqslant2$. To show the last point, first apply \amaurycorref{cor:ip1}:
if $x\leqslant\frac{1}{2\lambda}$, then $1-x+\frac{3}{4\lambda}\geqslant 1+\frac{1}{4\lambda}$,
thus $|\operatorname{nz}(x,\mu,\lambda)-x-\frac{2}{\lambda}|\leqslant \frac{2}{\lambda}e^{-\mu-1}$
Thus $\operatorname{nz}(x,\mu,\lambda)\geqslant \frac{2}{\lambda}(1-e^{-\mu-1})+x\geqslant \frac{1}{\lambda}$
since $1-e^{-\mu-1}\leqslant\frac{1}{2}$ and $x\geqslant0$. And for $x\geqslant\frac{1}{2\lambda}$,
by \amaurycorref{cor:ip1} we get that $\operatorname{nz}(x,\mu,\lambda)\geqslant x\geqslant\frac{1}{2\lambda}$
which shows the last point.

\noindent Then we show that:
\begin{itemize}
\item $\operatorname{cltan}\in\gval{\poly}$, is $\pi$-periodic and is an odd function.
\item For $\theta\in\left[-\frac{\pi}{2}+\frac{1}{\lambda},\frac{\pi}{2}-\frac{1}{\lambda}\right]$,
$|\operatorname{cltan}(\theta,\mu,\lambda)-\tan(\theta)|\leqslant e^{-\mu}$
\end{itemize}
First apply the above results to get that $\operatorname{nz}(\cos^2\theta,\mu+16\lambda^3,4\lambda^2)\geqslant\frac{1}{8\lambda^2}$.
It follows that $\operatorname{cltan}(\theta,\mu,\lambda)\leqslant
\frac{1}{\sqrt{\operatorname{nz}(\cos^2\theta,\mu+16\lambda^3,4\lambda^2)}}\leqslant\sqrt{8}\lambda$,
which is a polynomial in $\lambda$.
Since $\sin,\cos,\sg,\operatorname{nz}\in\gval{\poly}$, it follows that $\operatorname{clan}\in\gval{\poly}$.
The periodicity comes from the properties of sine and cosine, and the fact that $\sg$ is an odd function.
It is an odd function for similar reasons.
To show the second point, since it is periodic and odd, we can assume that $\theta\in\left[0,\frac{\pi}{2}-\frac{1}{\lambda}\right]$.
For such a $\theta$, we have that $\frac{\pi}{2}-\theta\geqslant\frac{1}{\lambda}$,
thus $\cos(\theta)\geqslant\sin(\frac{\pi}{2}-\theta)\geqslant\frac{1}{2\lambda}$ (use that
$\sin(u)\geqslant\frac{u}{2}$ for $0\leqslant u\leqslant\frac{\pi}{2}$). By
\lemref{lem:sg} we get that $|\sg(\cos\theta,\mu+3\lambda,2\lambda)-1|\leqslant e^{-\mu-3\lambda}$.
Also $\cos^2\theta\geqslant\frac{1}{4\lambda^2}$ thus by the above results we get that
$|\operatorname{nz}(\cos^2\theta,\mu+16\lambda^3,4\lambda^2)-\cos^2\theta|\leqslant e^{-\mu}$.
Using the fact that $|\frac{\sqrt{a}-\sqrt{b}}{\sqrt{a}}|\leqslant|a-b|$ for any $a>0$
and $b\in\R$, we get that
$\left|\frac{\sqrt{\operatorname{nz}(\cos^2\theta,\mu,4\lambda^2)}-|\cos\theta|}{\sqrt{\operatorname{nz}(\cos^2\theta,\mu+16\lambda^3,2\lambda)}}\right|
\leqslant|\operatorname{nz}(\cos^2\theta,\mu+16\lambda^3,4\lambda^2)-\cos^2\theta|
\leqslant \sqrt{8}\lambda e^{-\mu-16\lambda^3}$. Putting everything together,
using that $\cos\theta\geqslant\frac{1}{2\lambda}$ and $\operatorname{nz}(\cos^2\theta,\mu+16\lambda^3,2\lambda)\geqslant 8\lambda^2$,
we get that
\begin{align*}
|\operatorname{cltan}(\theta,\mu,\lambda)-\tan\theta|
    &=\left|\frac{\sin(\theta)\sg(\cos\theta,\mu+3\lambda,2\lambda)}{\sqrt{\operatorname{nz}(\cos^2\theta,\mu+16\lambda^3,4\lambda^2)}}-\frac{\sin\theta}{\cos\theta}\right|\\
    &\leqslant\left|\frac{\sin(\theta)(\sg(\cos\theta,\mu+3\lambda,2\lambda)-\sgn{\cos\theta}}{\sqrt{\operatorname{nz}(\cos^2\theta,\mu+16\lambda^3,4\lambda^2)}}\right|\\
    &\hphantom{h}+\left|\frac{\sin(\theta)\sgn{\cos\theta}}{\sqrt{\operatorname{nz}(\cos^2\theta,\mu+16\lambda^3,4\lambda^2)}}-\frac{\sin\theta}{\cos\theta}\right|\\
    &\leqslant\frac{|\sg(\cos\theta,\mu+3\lambda,2\lambda)-\sgn{\cos\theta}|}{\sqrt{\operatorname{nz}(\cos^2\theta,\mu+16\lambda^3,4\lambda^2)}}\\
    &\hphantom{h}+\left|\frac{1}{\sqrt{\operatorname{nz}(\cos^2\theta,\mu+16\lambda^3,4\lambda^2)}}-\frac{1}{|\cos\theta|}\right|\\
    &\leqslant\sqrt{8}\lambda e^{-\mu-3\lambda}
    +\frac{|\sqrt{\operatorname{nz}(\cos^2\theta,\mu+16\lambda^3,4\lambda^2)}-|\cos\theta||}
        {|\cos\theta\|\sqrt{\operatorname{nz}(\cos^2\theta,\mu+16\lambda^3,4\lambda^2)}}\\
    &\leqslant\sqrt{8}\lambda e^{-\mu-3\lambda}+2\lambda\cdot\sqrt{8}\lambda\cdot\sqrt{8}\lambda e^{-\mu-16\lambda^3}\\
    &\leqslant3\lambda e^{-\mu-3\lambda}+16\lambda^3 e^{-\mu-16\lambda^3}\\
    &\leqslant e^{-\mu}
\end{align*}
because $xe^{-x}\leqslant\tfrac{1}{2}$ for any $x\geqslant0$.

Let $n\in\N$ and $x=n+\delta\in[n-\frac{1}{2},n+\frac{1}{2}]$.
Since $\operatorname{cltan}$ is $\pi$-periodic,
$\rnd(x,\mu,\lambda)=n+\delta-\frac{1}{\pi}\arctan(\operatorname{cltan}(\pi\delta,\mu,\lambda))$.
Furthermore $\pi\delta\in[-\frac{\pi}{2},\frac{\pi}{2}]$ so $\cos(\pi\delta)\geqslant0$
and $\sgn{\sin(\pi\delta)}=\sgn{\delta}$.
Consequently, $\sg(\cos(\pi\delta),\mu+3\lambda,2\lambda)\in[0,1]$ by definition of $\sg$
and $\sqrt{\operatorname{nz}(\cos^2(\pi\delta),\mu+16\lambda^3,4\lambda^2)}>\sqrt{\cos^2(\pi\delta)}$ because $\ip{1}>0$.
Consequently, we get that $|\operatorname{cltan}(\pi\delta,\mu,\lambda)|\leqslant\frac{|\sin(\pi\delta)|}{\cos(\pi\delta)}$
and
$\sgn{\operatorname{cltan}(\pi\delta,\mu,\lambda)}=\sgn{\delta}$. Finally,
we can write
$\frac{1}{\pi}\arctan(\operatorname{cltan}(\pi\delta,\mu,\lambda))=\alpha$ with
$|\alpha|\leqslant|\frac{1}{\pi}\arctan(\tan(\pi\delta))|\leqslant|\delta|$
and $\sgn{\alpha}=\sgn{\delta}$ which shows that $|\rnd(x,\mu,\lambda)-n|\leqslant\delta\leqslant\frac{1}{2}$.

Finally we can show the result about $\rnd$: since $\operatorname{cltan}$ and $\tan$ are in $\gval{\poly}$, then
$\rnd\in\gval{\poly}$. Now consider $x\in\left[n-\frac{1}{2}+\frac{1}{\lambda},n+\frac{1}{2}-\frac{1}{\lambda}\right]$,
and let $\theta=\pi x-\pi n$. Then
$\theta\in\left[-\frac{\pi}{2}+\frac{\pi}{\lambda},\frac{\pi}{2}-\frac{\pi}{\lambda}\right]\subseteq
\left[-\frac{\pi}{2}+\frac{1}{\lambda},\frac{\pi}{2}-\frac{1}{\lambda}\right]$,
and since $\operatorname{cltan}$ is periodic, then
$\rnd(x,\mu,\lambda)=n+\frac{\theta}{\pi}-\frac{1}{\pi}\arctan(\operatorname{cltan}(\theta,\mu,\lambda)$.
Finally, using the results about $\operatorname{cltan}$ yields:
$|\rnd(x,\mu,\lambda)-n|=\frac{1}{\pi}|\theta-\arctan(\operatorname{cltan}(\theta,\mu,\lambda)|
=\frac{1}{\pi}|\arctan(\tan(\theta))-\arctan(\operatorname{cltan}(\theta,\mu,\lambda)|
\leqslant\frac{1}{\pi}|\tan(\theta)-\operatorname{cltan}(\theta,\mu,\lambda)|\leqslant \frac{e^{-\mu}}{\pi}\leqslant e^{-\mu}$
since $\arctan$ is a $1$-Lipschitz function.
\end{proof}

\subsection{Absolute value, maximum and norm}

A very common operation is to compute the absolute value of
a number. Of course this operation is not generable because it is not even
differentiable. However, a good enough approximation can be built. In particular,
this approximation has several key features: it is non-negative and it is an over-approximation.
We can then use it to build an approximation of the max function and the infinite norm.

\begin{definition}[Absolute value function]
For any $x\in\R$ and $\mu,\lambda>0$ define:
\[\abs(x,\mu,\lambda)=\frac{1}{1+\lambda\mu}\ln(2\cosh((1+\lambda\mu)x))\]
\end{definition}

\begin{lemma}[Absolute value]\label{lem:abs}
For any $x\in\R$ and $\mu,\lambda>0$ we have
\[|x|\leqslant\abs(x,\mu,\lambda)\leqslant|x|+\min\left(\frac{1}{1+\lambda\mu},e^{-|x|\lambda\mu}\right).\]
So in particular, if $|x|\geqslant\lambda^{-1}$ then $|x|\leqslant\abs(x,\mu,\lambda)\leqslant|x|+e^{-\mu}$.
Furthermore $\abs\in\gval{\poly}$ and is an even function.
\end{lemma}

\begin{proof}
Since $\cosh$ is an even function, we immediately get that $\abs$ is even.
Let $x\geqslant0$ and $\mu,\lambda>0$. Since $2\cosh(u)\geqslant e^{u}$,
it trivially follows that $\abs(x,\mu,\lambda)\geqslant\frac{1}{1+\lambda\mu}(1+\lambda\mu)x\geqslant x$.
Also $\ln(2\cosh(u))=\ln(e^{u}(1+e^{-2u}))=u+\ln(1+e^{-2u})\leqslant u+e^{-2u}$
so it follows that $\abs(x,\mu,\lambda)\leqslant x+\frac{1}{1+\lambda\mu}e^{-2(1+\lambda\mu)x}\leqslant x+e^{-x\lambda\mu}$.
Furthermore, $\frac{\partial\abs}{\partial x}(x,\mu,\lambda)=\tanh((1+\lambda\mu)x)$
which shows that $x\mapsto\abs(x,\mu,\lambda)-x$ is decreasing and positive over $[0,+\infty[$ and
thus has its maximum $\abs(0,\mu,\lambda)=\frac{1}{1+\mu\lambda}$ attained at $0$.
Since $\big(\ln(2\cosh(u))\big)'=\tanh(u)$, $\tanh\in\gval{\poly}$ and $\ln(2\cosh(u))$
is bounded by $|u|+1$, we get that $\big(u\mapsto\ln(2\cosh(u))\big)\in\gval{\poly}$
by applying Corollary \ref{cor:gpac_ext_ivp_stable}. It follows that $\abs\in\gval{\poly}$ using the usual lemmas.
\end{proof}

\begin{definition}[Max/Min function]
For any $x,y\in\R$ and $\mu,\lambda>0$ define:
\[\mx(x,y,\mu,\lambda)=\frac{y+x+\abs(y-x,\mu,\lambda)}{2}\qquad
\mn(x,y,\mu,\lambda)=x+y-\mx(x,y,\mu,\lambda).\]
For any $x\in\R^n$ and $\delta\in]0,1]$ define:
\[\mx_\delta(x)=\mx(x_1,\mx(\ldots,\mx(x_{n-1},x_n,1,(n\delta)^{-1})\ldots)).\]
\end{definition}

\begin{lemma}[Max/Min function]\label{lem:max}
For any $x,y\in\R$ and $\lambda,\mu>0$ we have:
\[\max(x,y)\leqslant\mx(x,y,\mu,\lambda)\leqslant\max(x,y)+\min\left(\frac{1}{1+\lambda\mu},e^{-|x-y|\lambda\mu}\right)\]
and
\[\min(x,y)-\min\left(\frac{1}{1+\lambda\mu},e^{-|x-y|\lambda\mu}\right)\leqslant\mn(x,y,\mu,\lambda)\leqslant\min(x,y)\]
So in particular, if $|x-y|\geqslant\lambda^{-1}$ then $\max(x,y)\leqslant\mx(x,y,\mu,\lambda)\leqslant\max(x,y)+e^{-\mu}$
and $\min(x,y)-e^{-\mu}\leqslant\mn(x,y,\mu,\lambda)\leqslant\min(x,y)$.
Furthermore $\mx,\mn\in\gval{\poly}$.
For any $x\in\R^n$ and $\delta\in]0,1]$ we have:
\[\max(x_1,\ldots,x_n)\leqslant\mx_\delta(x)\leqslant\max(x_1,\ldots,x_n)+\delta\]
Furthermore $\mx_\delta\in\gval{\poly}$.
\end{lemma}

\begin{proof}
By \lemref{lem:abs}, $|y-x|\leqslant\abs(y-x,\mu,\lambda)\leqslant |y-x|+\min\left(\frac{1}{1+\lambda\mu},e^{-|x-y|\lambda\mu}\right)$
and the result follows because $\max(x,y)=\frac{y+x+|y-x|}{2}$. The result on $\mn$ follows the one on $\mx$.
Finally $\mx,\mn\in\gval{\poly}$ from \lemref{lem:gpac_ext_class_stable}.

Observe that $\max(x)\leqslant\mx_\delta(x)$ is trivial by definition. The other inequality
is a simple calculus based on $\max(x,y,\mu,\lambda)\leqslant\max(x,y)+\frac{1}{1+\mu\lambda}$:
\[\mx_\delta(x)\leqslant\max(x)+n\frac{1}{1+(n\delta)^{-1}}\leqslant\max(x)+\delta.\]
Note that strictly speaking, for $\mx_\delta\in\gval[\K]{\poly}$ we need that $\delta\in\K$
or use a smaller $\delta'$ in $\K$ which is always possible.
\end{proof}

\begin{definition}[Norm function]
For any $x\in\R^n$ and $\delta\in]0,1]$ define:
\[\norm_{\infty,\delta}(x)=\mx_{\delta/2}(\abs_{\delta/2}(x_1),\ldots,\abs_{\delta/2}(x_n))\]
where $\abs_\delta(x)=\mx_{\delta}(x,-x)$.
\end{definition}

\begin{lemma}[Norm function]\label{lem:norm}
For any $x\in\R^n$ and $\delta\in]0,1]$ we have:
\[\infnorm{x}\leqslant\norm_{\infty,\delta}(x)\leqslant\infnorm{x}+\delta\]
Furthermore $\norm_{\infty,\delta}\in\gval{\poly}$.
\end{lemma}

\begin{proof}
Apply \lemref{lem:abs} and \lemref{lem:max}.
\end{proof}

\begin{figure}
\begin{center}
\begin{tikzpicture}[domain=0:4,samples=200,scale=1.5]
\draw[very thin,color=gray] (-0.1,-0.1) grid (4.1,1.1);
\draw[->] (0,0.1) -- (0,1.1);
\draw[->] (-0.1,0) -- (4.1,0) node[right] {$x$};
\draw[color=red] plot[id=fn_lxh] function{\fnlxh{1}{3}{x}{5}{1}};
\draw[color=blue] plot[id=fn_hxl] function{\fnhxl{1}{2}{x}{5}{1}};
\end{tikzpicture}
\end{center}
\caption{Graph of $\lxh_{[1,3]}$ and $\hxl_{[1,2]}$}
\label{fig:lxh_hxl}
\end{figure}
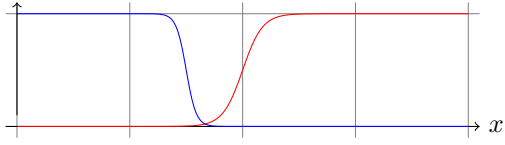

\subsection{Switching functions}

An important construct in digital computation is the ``if ... then ... else ...'' construct,
which allows us to switch between two different behaviours. Again, this cannot be done
exactly with a GPAC since GPACs cannot generate discrete functions
and we need
something which acts like a select function, which can pick between two values
depending on how a third value compares to a threshold. The problem is that this operation
is not continuous, and thus not generable. But such a select function can be approximated by a GPAC. As a good first step, we build so-called
``low-X-high'' and ``high-X-low'' functions which act as a switch between $0$ (low) and a value (high).
Around the threshold will be an small uncertainty zone (X) where the exact value
cannot be predicted. See \figref{fig:lxh_hxl} for a graphical representation.

\begin{definition}[``low-X-high'' and ``high-X-low'']\label{def:lxh_hxl}
Let $I=[a,b]$ with $b>a$, $t\in\R$, $\mu\in\R$, $x\in\R$, $\nu=\mu+\ln(1+x^2)$, $\delta=\frac{b-a}{2}$
and define:
\[\lxh_I(t,\mu,x)=\ip{1}\left(t-\frac{a+b}{2}+1,\nu,\frac{1}{\delta}\right)x
\qquad\hxl_I(t,\mu,x)=\ip{1}\left(\frac{a+b}{2}-t+1,\nu,\frac{1}{\delta}\right)x\]
\end{definition}

\begin{lemma}[``low-X-high'' and ``high-X-low'']\label{lem:lxh_hxl}
Let $I=[a,b]$, $\mu\in\Rp$, then $\forall t,x\in\R$:
\begin{itemize}
\item $\exists\phi_1,\phi_2$ such that $\lxh_I(t,\mu,x)=\phi_1(t,\mu,x)x$ and $\hxl_I(t,\mu,x)=\phi_2(t,\mu,x)x$
\item if $t\leqslant a, |\lxh_I(t,\mu,x)|\leqslant e^{-\mu}$ and $|x-\hxl_I(t,\mu,x)|\leqslant e^{-\mu}$
\item if $t\geqslant b, |x-\lxh_I(t,\mu,x)|\leqslant e^{-\mu}$ and $|\hxl_I(t,\mu,x)|\leqslant e^{-\mu}$
\item in all cases, $|\lxh_I(t,\mu,x)|\leqslant|x|$ and $|\hxl_I(t,\mu,x)|\leqslant|x|$
\end{itemize}
Furthermore, $\lxh_I,\hxl_I\in\gval{\poly}$.
\end{lemma}

\begin{proof}
By symmetry, we only prove it for $\lxh$.
This is a direct consequence of \amaurycorref{cor:ip1} and the fact that $|x|\leqslant e^{\ln(1+x^2)}$.
Indeed if $t\leqslant a$ then $t-\frac{a+b}{2}+1\leqslant 1-\delta$ thus
$|\lxh_I(t,\nu,x)|\leqslant |x|e^{-\nu}\leqslant e^{-\mu}$.
Similarly if $t\geqslant b$ then $t-\frac{a+b}{2}+1\geqslant 1+\delta$ and we get a similar result.
Apply \lemref{lem:gpac_ext_class_stable} multiple times to see that they are belong to $\gval{\poly}$.
\end{proof}

\subsection{GPAC approximation}\label{sec:zoo_generic}

The examples of the previous section all share an interesting common pattern, which we formalise with the
definition below. In this section, $\K$ can be any generable field\footnote{See Section \ref{sec:generable_field} for more details.}.

\begin{definition}[GPAC approximation]
Let $I$ be an open and connected subset of $\R^m$, $\Gamma\subseteq I$ a subset of $I$ of \emph{exceptions} and $f:I\rightarrow\R^m$.
We say that $f$ is \emph{GPAC-approximable over $I$ but $\Gamma$} if there exists $g\in\gval[\K]{\poly}$
such that for any $x\in I$ and $\mu,\lambda>0$ we have
\[\infnorm{f(x)-g(x,\mu,\lambda)}\leqslant e^{-\mu}\quad\text{ if }\quad d(x,\Gamma)\geqslant\lambda^{-1},\]
where $d(x,\Gamma)$ denotes the distance between $x$ and $\Gamma$ (for the infinite norm).
\end{definition}

The set $\Gamma$ of points where the approximation fails will typically be discrete, finite or even empty.
If $\Gamma$ is empty, we do not mention it and say $f$ is GPAC-approximable.
Intuively, $g$ provides an \emph{effective}, uniform and arbitrary good approximation of $f$,
except on a set that can be made ``arbitrary small''. We cannot quantify how small the set
of exception is in general, since the definition allows for pathological cases such as $\Gamma=I$
or $\Gamma=I\cap \Q^m$. However, in case where $\Gamma$ is discrete, a condition met by all examples in this
paper, for any compact set $K$, the measure of exception set $\{d(x,\Gamma)\leqslant\lambda^{-1}\}\cap K$
converges to $0$ as $\lambda$ tends to infinity.

Note that our notion of approximation is not really related to classical approximation
theory, by a sequence of functions for example. Indeed, in the definition, the
same function $g$ is used for all $\mu$ and $\lambda$, which creates a lot of constraints
since $g$ is generable, i.e.~it satisfies a polynomial partial differential equation.
Informally, one can think of $g$ as a ``template'' with parameters $\mu$ and $\lambda$ that we can tweak to get closer
and closer to $f$ but the shape itself of the template is fixed once and for all.

It appears that there is an interesting trade-off between the bound $\mtt{sp}$ on the norm of $g$ (i.e. $g\in\gval{\mtt{sp}}$)
and the quality of the approximation. Indeed, if $\mtt{sp}$ is chosen to be a polynomial,
we can seemingly achieve an exponential error bound ($e^{-\mu}$) but only an inverse distance from $\Gamma$ ($1/\lambda$)
for interesting functions.
For simplicity, we only consider polynomially bounded generable functions is this definition.

Note that the definition does not mandate that $f$ be continuous
and indeed it needs not be. For example, Lemma \ref{lem:rnd} proves that the rounding
function is GPAC-approximable over $\R$ but $\frac{1}{2}+\Z$. More generally, the discontinuity
points will always belong to $\Gamma$.

In this section, we give several examples of classes of functions that can be approximated
as described above.

\begin{lemma}[Basic approximable functions]
Any generable function is approximable on its domain of definition. If $f$ and $g$ are GPAC-approximable
over $X$ but $\Gamma_f$ and $\Gamma_g$ respectively, then $f\pm g$ and $fg$ are GPAC-approximable
over $X$ but $\Gamma_f\cup\Gamma_g$.
\end{lemma}

\begin{proof}
Any generable function trivially satisfies the definition using itself as an approximation.
If $f$ is approximated by $F$ and $g$ by $G$ then for any $\mu,\lambda>0$ and $x\in X$ such that
$d(x,\Gamma_f\cup\Gamma_g)\geqslant\lambda^{-1}$:
\[\infnorm{f(x)+g(x)-F(x,\mu+1,\lambda)-G(x,\mu+1,\lambda)}\leqslant2e^{-\mu-1}\leqslant e^{-\mu}.\]
Thus $\lambdafun{x,\mu,\lambda}{F(x,\mu+1,\lambda)+G(x,\mu+1,\lambda)}$ approximate $f+g$
over $X$ but $\Gamma_f\cup\Gamma_g$.

The case of the multiplication is similar but slightly more involved.
Define for any $x\in X$ and $\mu,\lambda>0$:
\[H(x,\mu,\lambda)=\underbrace{F\big(x,\mu+2+\norm_{\infty,1}(G(x,1,\lambda)),\lambda\big)}_{:=\tilde{f}(x,\mu,\lambda)}
\underbrace{G\big(x,\mu+3+\norm_{\infty,1}(F(x,1,\lambda)),\lambda\big)}_{:=\tilde{g}(x,\mu,\lambda)}.\]
It will be useful to recall that $\infnorm{x}\leqslant\norm_{\infty,1}(x)$ thanks to Lemma \ref{lem:norm}.
Let $\mu,\lambda>0$ and $x\in X$ such that $d(x,\Gamma_f\cup\Gamma_g)\geqslant\lambda^{-1}$. Note
that since $\infnorm{f(x)-F(x,1,\lambda)}\leqslant e^{-1}$ then $\infnorm{F(x,1,\lambda)}\geqslant\infnorm{f(x)}-1$.
Similarly, $\infnorm{\tilde{g}(x,\mu,\lambda)-G(x,1,\lambda)}\leqslant e^{-1}+e^{-\mu}$ thus
$\infnorm{G(x,1,\lambda)}\geqslant\infnorm{\tilde{g}(x,\mu,\lambda)}-2$. Finally check that $x\mapsto x e^{-x}$
is globally bounded by $1$. Thus we have:
\begin{align*}
\infnorm{f(x)g(x)-H(x,\mu,\lambda)}
    &\leqslant \infnorm{f(x)}\infnorm{g(x)-\tilde{g}(x,\mu,\lambda)}\\
    &\;+\infnorm{f(x)-\tilde{f}(x,\mu,\lambda)}\infnorm{\tilde{g}(x,\mu,\lambda)}\\
    &\leqslant\infnorm{f(x)}e^{-\mu-2-\norm_{\infty,1}(F(x,1,\lambda))}\\
    &\;+e^{-\mu-3-\norm_{\infty,1}(G(x,1,\lambda))}\infnorm{\tilde{g}(x,\mu,\lambda)}\\
    &\leqslant\infnorm{f(x)}e^{-\mu-1-\infnorm{f(x)}}
    +e^{-\mu-1-\infnorm{\tilde{g}(x,\mu,\lambda)}}\infnorm{\tilde{g}(x,\mu,\lambda)}\\
    &\leqslant2e^{-\mu-1}\leqslant e^{-\mu}.
\end{align*}
This shows that $H$ approximates $fg$ over $x$ but $\Gamma_f\cup\Gamma_g$. The fact
that $H\in\gval{\poly}$ follows from the hypothesis on $F$ and $G$ and Lemma \ref{lem:gpac_ext_class_stable}.
\end{proof}

\begin{theorem}[Piecewise approximability]\label{th:piecewise_gpac}
Let $-\infty\leqslant a_0<a_1<\ldots<a_{k+1}\leqslant+\infty$ and $f:]a_0,a_{k+1}[\rightarrow\R$.
Assume that for each $i\in\{0,\ldots,k\}$, $f$ is GPAC-approximable over $]a_i,a_{i+1}[$ but $\Gamma_i$.
Further assume that all finite $a_i$ belong to $\K$. Then $f$ is GPAC-approximable over $]a_0,a_{k+1}[$
but $\{a_1,\ldots,a_k\}\cup\bigcup_{i=0}^k\Gamma_i$.
\end{theorem}

\begin{proof}
Without loss of generality, we can assume that $f$ is defined over $\R$. Indeed if $f$ is only
defined over $[a,b]$, $[a,+\infty[$ or $]-\infty,b]$, we can add an extra infinite interval
over which $f$ is constantly equal to $0$. The resulting $g$ for this extended $f$ satisfies the
definition over the original domain of definition of $f$.

We now assume that $a_0=-\infty$ and $a_{k+1}=+\infty$. Let $\tilde{f}_i\in\gval{\poly}$
be the GPAC-approximation of $f$ over $]a_i,a_{i+1}[$ but $\Gamma_i$, for $i\in\{0,\ldots,k\}$.
There is a subtle issue at this point: a priori $\tilde{f}_i$ is only defined over
$]a_i,a_{i+1}[\times]0,+\infty[^2$. We will show that $\tilde{f}_i$ can be assumed to
be defined over $\R\times]0,+\infty[^2$ and we defer of proof of this fact to end of this proof.
Define for any $x\in\R$, $\mu\geqslant0$ and $\lambda>0$:
\[g(x,\mu,\lambda)=\tilde{f}_0(x,\nu,\lambda)+
\sum_{i=1}^k\lxh_{[-1,1]}\big((x-a_i)\lambda,\nu,\tilde{f}_i(x,\nu,\lambda)-\tilde{f}_{i-1}(x,\nu,\lambda)\big)\]
where $\nu=\mu+k+1$. First note that $g\in\gval[\K]{\poly}$ because it is a finite
sum of generable functions in $\gval{\poly}$, and the endpoints of the intervals belong to $\K$.
Define $\Gamma=\{a_1,\ldots,a_k\}\cup\bigcup_{i=0}^k\Gamma_i$. Let $\mu,\lambda>0$ and $x\in\R$
be such that $d(x,\Gamma)\geqslant\lambda^{-1}$. It follows that $a_i+\lambda^{-1}\leqslant x\leqslant a_{i+1}-\lambda^{-1}$
for some $i\in\{0,\ldots,k\}$. Let $j\in\{0,\ldots,k\}$ and apply Lemma \ref{lem:lxh_hxl}
to get that $|\lxh_{[-1,1]}((x-a_j)\lambda,\nu,X)|\leqslant e^{-\nu}$ if $j\geqslant i+1$
and $|\lxh_{[-1,1]}((x-a_j)\lambda,\nu,X)-X|\leqslant e^{-\nu}$ if $j\leqslant i$.
It follows that:
\begin{align*}
\left|g(x,\mu,\lambda)-f(x)\right|
&\leqslant\left|g(x,\mu,\lambda)-\tilde{f}_i(x,\nu,\lambda)\right|+e^{-\nu}\\
    &=\left|g(x,\mu,\lambda)-\tilde{f}_0(x,\nu,\lambda)-\sum_{j=1}^i\left(\tilde{f}_i(x,\nu,\lambda)-\tilde{f}_{i-1}(x,\nu,\lambda)\right)\right|\\
    &\leqslant\sum_{j=1}^i\Big(\lxh_{[-1,1]}\big((x-a_i)\lambda,\nu,\tilde{f}_i(x,\nu,\lambda)-\tilde{f}_{i-1}(x,\nu,\lambda)\big)\\
    &\hspace{1.5cm}-\left(\tilde{f}_i(x,\nu,\lambda)-\tilde{f}_{i-1}(x,\nu,\lambda)\right)\Big)+e^{-\nu}\\
    &\leqslant (k+1)e^{-\nu}\leqslant e^{-\mu}.
\end{align*}
This concludes the proof that $f$ is approximate by $g$ over $\R$ but $\Gamma$.
It remains to show that, indeed, each $\tilde{f}_i$ can be assumed to be defined over $\R$.
We show this in full-generality for intervals.

Let $f:]a,b[\rightarrow\R$ and $\tilde{f}:]a,b[\times]0,+\infty[^2$ a GPAC-approximation
of $f$. Let $\mtt{sp}$ be a polynomial such that $\tilde{f}\in\gval{\mtt{sp}}$.
Apply Proposition \ref{prop:generable_mod_cont} to $\tilde{f}$ to get a polynomial $q$. Recall
that $q$ acts as a modulus of continuity:
\[\left|\tilde{f}(x,\mu,\lambda)-f(y,\mu,\lambda)\right|\leqslant|x-y|q(\mtt{sp}(\max(|x|,|y|,\mu,\lambda)))\]
for any $x,y\in]a,b[$ and $\mu,\lambda>0$. Let $p\in\K[\R]$ be a nondecreasing polynomial
such that $p(x)\geqslant q(\mtt{sp}(x))$ for all $x\geqslant0$. Define for any $x\in\R$ and
$\mu,\lambda>0$:
\[\clamp(x,\mu,\lambda)=\mx(a+\theta^{-1},\mn(x,b-\theta^{-1},\mu+1,\theta),\mu+1,\theta)\]
where $\delta=b-a$ and $\theta=2\lambda+(2\delta)^{-1}$. Observe that $\clamp$ satisfies three key properties:
\begin{itemize}
\item $\clamp(x,\mu,\lambda)\in]a,b[$ for all $x\in\R$ and $\mu,\lambda>0$: by Lemma \ref{lem:max},
$\clamp(x,\mu,\lambda)\geqslant a+\theta^{-1}>a$. On the other hand,
$\clamp(x,\mu,\lambda)\leqslant\max(a+\theta^{-1},\mn(x,b,\mu+1,\theta))+\frac{1}{1+(1+\mu)\theta}$ but
$\mn(\mn(x,b-\theta^{-1},\mu+1,\theta))\leqslant b-\theta^{-1}$ so
$\clamp(x,\mu,\lambda)\leqslant\max(a+\theta^{-1},b-\theta^{-1})+\frac{1}{1+(1+\mu)\theta}$.
Note that $\theta>(2\delta)^{-1}$ so $a+\theta^{-1}<b-\theta^{-1}$. Consequently
$\clamp(x,\mu,\lambda)\leqslant b-\theta^{-1}+\frac{1}{1+(1+\mu)\theta}<b$.
\item if $a+\lambda^{-1}\leqslant x\leqslant b-\lambda^{-1}$ then $|\clamp(x,\mu,\lambda)-x|\leqslant e^{-\mu}$:
if $a+\lambda^{-1}\leqslant x$ then $x-(a+\theta^{-1})-\theta^{-1}\geqslant\lambda^{-1}-2\theta^{-1}\geqslant0$
so $|\clamp(x,\mu,\lambda)-\mn(x,b-\theta^{-1},\mu+1,\theta)|\leqslant e^{-\mu-1}$. Similarly, $x\leqslant b-\lambda^{-1}$
implies that $x\leqslant(b-\theta^{-1})-\theta^{-1}$ so $|\mn(x,b-\theta^{-1},\mu+1,\theta)-x|\leqslant e^{-\mu-1}$.
It follows that $|\clamp(x,\mu,\lambda)-x|\leqslant2e^{-\mu-1}\leqslant e^{-\mu}$.
\item $\clamp\in\gval{\poly}$: use Lemma \ref{lem:max} and the usual arithmetic lemmas.
Note that it works because $\lambda\mapsto(2\lambda+(2\delta)^{-1})^{-1}$ belongs
to $\gval{\poly}$ for any fixed $\delta$.
\end{itemize}
We can now use $\clamp$ to make sure the argument of $\tilde{f}$ is always within
the domain of definition $]a,b[$, and make sure that it is a good enough approximation
using the modulus of continuity. Define for any $x\in\R$ and $\mu,\lambda>0$:
\[\tilde{F}(x,\mu,\lambda)=\tilde{f}(\clamp(x,\mu+1+p(1+\norm_{\infty,1}(x,\mu,\lambda)),\lambda),\mu+1,\lambda)\]
Clearly $\tilde{F}\in\gval{\poly}$. Let $\mu,\lambda>0$ and $x\in]a,b[$ such that
$d(x,\Gamma\cup\{a,b\})\geqslant\lambda^{-1}$. It follows from the results above that:
\begin{align*}
\left|f(x)-\tilde{F}(x,\mu,\lambda)\right|
    &\leqslant\left|f(x)-\tilde{f}(x,\mu+1,\lambda)\right|+\left|\tilde{F}(x,\mu,\lambda)-\tilde{f}(x,\mu+1,\lambda)\right|\\
    &\leqslant e^{-\mu-1}+\big|x-\clamp(x,\mu+1+p(1+\norm_{\infty,1}(x,\mu,\lambda)),\lambda)\big|\\
    &\hspace{.5cm}\times p\left(\max(|x|,\big|\clamp(x,\mu+1+p(1+\norm_{\infty,1}(x,\mu,\lambda)),\lambda)\big|,\mu+1,\lambda)\right)\\
    &\leqslant e^{-\mu-1}+e^{-\mu-1-p(1+\norm_{\infty,1}(x,\mu,\lambda))}\\
    &\hspace{.5cm}\times p\left(\max(|x|,|x|+e^{-\mu-1-p(1+\norm_{\infty,1}(x,\mu,\lambda))},\mu+1,\lambda)\right)\\
    &\leqslant e^{-\mu-1}+e^{-\mu-1-p(\max(1+|x|,\mu+1,\lambda))}p(\max(|x|,|x|+1,\mu+1,\lambda))\\
    &\leqslant2e^{-\mu-1}\leqslant e^{-\mu}
\end{align*}
\end{proof}

\begin{theorem}[Periodic approximability]
Let $f:\R\rightarrow\R$ be a $\tau$-periodic function. Assume that there exists $a,b\in\K$
such that $b-a=\tau$ and $f$ is GPAC-approximable over $]a,b[$ but $\Gamma$.
Then $f$ is GPAC-approximable over $\R$ but $(\Gamma\cup\{a,b\})+\tau\Z$.
\end{theorem}

\begin{proof}
First note that we can assume that $a+b=0$: define $g(x)=f(x+\delta)$ where $\delta=\frac{a+b}{2}$, take
a GPAC-approximation $\tilde{f}$ of $f$ over $]a,b[$ but $\Gamma$. Observe that
$\tilde{g}(x,\mu,\lambda)=\tilde{f}(x+\delta,\mu,\lambda)$ provides an approximation of $g$
over $]a-\delta,b-\delta]$ but $\Gamma-\delta$. Then $f$ is approximable over $\R$ but $(\Gamma\cup\{a,b\})+\tau\Z$
if and only if $g$ is approximable over $\R$ but $((\Gamma-\delta)\cup\{a-\delta,b-\delta\})+\tau\Z$.
Now observe that $(a-\delta)+(b-\delta)=a+b-2\delta=0$.

For a similar reason, we can assume that $\tau=1$ by rescaling $x$. It follows
that we can assume that $a=-1/2$ and $b=1/2$. Let $\tilde{f}$ be a GPAC-approximation
of $f$ over $]\tfrac{-1}{2},\tfrac{1}{2}[$ but $\Gamma$. We use the same trick as
in Theorem \ref{th:piecewise_gpac} to ensure that $\tilde{f}$ is defined over $\R\times]0,+\infty[^2$.
Let $\mtt{sp}$ be a polynomial such that $\tilde{f}\in\gval{\mtt{sp}}$.
Apply Proposition \ref{prop:generable_mod_cont} to $\tilde{f}$ to get a polynomial $q$. Recall
that $q$ acts as a modulus of continuity:
\[\left|\tilde{f}(x,\mu,\lambda)-f(y,\mu,\lambda)\right|\leqslant|x-y|q(\mtt{sp}(\max(|x|,|y|,\mu,\lambda)))\]
for any $x,y\in]a,b[$ and $\mu,\lambda>0$. Let $p\in\K[\R]$ be a nondecreasing polynomial
such that $p(x)\geqslant q(\mtt{sp}(x))$ for all $x\geqslant0$. Define for any $x\in\R$ and $\mu,\lambda>0$:
\[\tilde{F}(x,\mu,\lambda)=\tilde{f}(x-\rnd(x,\mu+1+p(1+\norm_{\infty,1}(\mu,\lambda)),\lambda),\mu+1,\lambda)\]
Clearly $\tilde{F}\in\gval{\poly}$. Let $\mu,\lambda>0$ and $x\in]a,b[$ such that
$d(x,(\Gamma\cup\{a,b\})+\tau\Z)\geqslant\lambda^{-1}$. It follows that there exists $n\in\Z$
such that $x=n+u$ where $u\in]\tfrac{-1}{2}+\lambda^{-1},\tfrac{1}{2}-\lambda^{-1}[$
and $d(u,\Gamma)\geqslant\lambda^{-1}$. Apply Lemma \ref{lem:rnd} to get that
$|\rnd(x,\mu+1+p(1+\norm_{\infty,1}(\mu,\lambda)),\lambda)-n|\leqslant e^{-\mu-1-p(1+\norm_{\infty,1}(\mu,\lambda))}$
so in particular $|x-\rnd(x,\mu+1+p(1+\norm_{\infty,1}(\mu,\lambda)),\lambda)-u|\leqslant e^{-\mu-1-p(1+\norm_{\infty,1}(\mu,\lambda))}$.
In particular, $|x-\rnd(x,\mu+1+p(1+\norm_{\infty,1}(\mu,\lambda)),\lambda)|\leqslant1$.
It follows that:
\begin{align*}
\left|f(x)-\tilde{F}(x,\mu,\lambda)\right|
    &\leqslant\left|f(x)-\tilde{f}(u,\mu+1,\lambda)\right|+\left|\tilde{F}(x,\mu,\lambda)-\tilde{f}(u,\mu+1,\lambda)\right|\\
    &\leqslant \left|f(x-n)-\tilde{f}(u,\mu+1,\lambda)\right|\\
    &\hspace{.5cm}+\big|x-\rnd(x,\mu+1+p(1+\norm_{\infty,1}(\mu,\lambda)),\lambda)-u\big|\\
    &\hspace{.5cm}\times p\left(\max(|u|,\big|x-\rnd(x,\mu+1+p(1+\norm_{\infty,1}(\mu,\lambda)),\lambda)\big|,\mu+1,\lambda)\right)\\
    &\leqslant e^{-\mu-1}+e^{-\mu-1-p(1+\norm_{\infty,1}(\mu,\lambda))}p\left(\max(1,1,\mu+1,\lambda)\right)\\
    &\leqslant e^{-\mu-1}+e^{-\mu-1-p(\max(1,\mu+1,\lambda))}p(\max(1,\mu+1,\lambda))\\
    &\leqslant2e^{-\mu-1}\leqslant e^{-\mu}
\end{align*}

\end{proof}

\section{Generable fields}\label{sec:generable_field}

In Section \ref{sec:generable}, we introduced the notion of \emph{generable field},
which are fields with an additional stability property. We used this notion to
ensure that the class of functions we built is closed under composition.
It is well-known that if we allow any choice of constants in our computation,
we will gain extra computational power because of uncomputable real numbers.
For this reason, it is wise to make sure that we can exhibit at least one generable
field consisting of computable real numbers only, and possibly only polynomial
time computable numbers in the sense of computable analysis \cite{newcomputationalparadigms}.

Intuitively, we are looking for a (the) smallest generable field, call it $\Rgen$,
in order to minimize the computation power of the real numbers it contains.
The rest of this section is dedicated to the study of this field. We first recall
Definition \ref{def:generable_field}.

\begin{definition}[Generable field]
A field $\K$ is \emph{generable} if and only if $\Q\subseteq\K$ and for any
$\alpha\in\K$, and $(f:\R\rightarrow\R)\in\gval[\K]{}$,
$f(\alpha)\in\K$.
\end{definition}

\subsection{Extended stability}\label{sec:gen_field_ext_stab}

By definition of a generable field, $\K$ is preserved by unidimensional generable
functions. An interesting question is whether $\K$ is also
preserved by multidimensional functions. This is not immediate because because of several
key differences in the definition of multidimensional generable functions.
We first recall a folklore topology lemma.

\begin{lemma}[Offset of a compact set]\label{lem:offset_compact_set}
Let $X\subseteq U\subseteq\R^n$ where $U$ is open and $X$ is compact. Then there
exists $\varepsilon>0$ such that $X_\varepsilon\subseteq U$ where the $\varepsilon$-offset
of $X$ is defined by $X_\varepsilon=\bigcup_{x\in X}\openball{x}{\varepsilon}$.
\end{lemma}

\begin{proof}This is a very classical result: let $F=\R^n\setminus U$, then $F$
is closed so the distance function\footnote{We always use the infinite norm $\infnorm{\cdot}$
in this paper but it works for any distance} $d_F$ to $F$ is continuous. Since $X$ is compact,
$d_F(X)$ is a compact subset of $\Rp$, and $d_F(X)$ is nowhere $0$ because $X\subseteq U\subseteq F$
where $U$ is open. Consequently $d_F(X)$ admits a positive minimum $\varepsilon$.
Let $x\in X_\varepsilon$, then $\exists y\in X$ such that $\infnorm{x-y}<\varepsilon$,
and by the triangle inequality, $\varepsilon\leqslant d_F(y)\leqslant \infnorm{x-y}+d_F(x)$
so $d_F(x)>0$ which means $x\notin F$, in other words $x\in U$.
\end{proof}

\begin{lemma}[Polygonal path connectedness]\label{lem:polygonal_connected}
An open, connected subset $U$ of $\R^n$ is always \emph{polygonal-path-connected}:
for any $a,b\in U$, there exists a polygonal path\footnote{A polygonal path
is a connected sequence of line segments} from $a$ to $b$ in $U$. Furthermore, we can take
all intermediate vertices in $\Q^n$.
\end{lemma}

\begin{proof}This is a textbook property, e.g. Theorem~3-5 in \cite{HockingYoungTopology}.
\end{proof}

\begin{proposition}[Generable path connectedness]\label{prop:connected_is_generable_connected}
An open, connected subset $U$ of $\R^n$ is always \emph{generable-path-connected}:
for any $a,b\in U\cap\K^n$, there exists $(\phi:\R\rightarrow U)\in\gval[\K]{}$ such that
$\phi(0)=a$ and $\phi(1)=b$.
\end{proposition}

\begin{proof}
Let $a,b\in U\cap\K^n$ and apply \lemref{lem:polygonal_connected} to get a polygonal path
$\gamma:[0,1]\rightarrow U$ from $a$ to $b$. We are going to build a highly smoothed
approximation of $\gamma$. This is usually done using bump functions
but bump functions are not analytic, which complicates the matter. Furthermore,
we need to build a path which domain of definition is $\R$, although this will be
a minor annoyance only. We ignore the case where $a=b$ which is trivial and focus
on the case where $a\neq b$.

Let $X=\gamma([0,1])$ which is a compact connected set. Apply \lemref{lem:offset_compact_set}
to get $\varepsilon>0$ such that $X_\varepsilon\subseteq U$. Without loss of
generality, we can assume that $\varepsilon\in\Q$ so that it is generable.

Assume for a moment that $\gamma$ is trivial, that is $\gamma$ is a line segment
from $a$ to $b$. Let $\alpha\in\N\subseteq\K$ such that
$\frac{1}{\tanh(\alpha)}\leqslant1+\frac{2\varepsilon}{\infnorm{b-a}}$.
It exists because $\frac{1}{\tanh(x)}\xrightarrow[x\rightarrow\infty]{}1$.
Define $\phi(t)=a+\frac{1+\mu(t)}{2}(b-a)$ where
$\mu(t)=\frac{\tanh((2t-1)\alpha)}{\tanh(\alpha)}$. One can check that
$\mu$ is an increasing function and that $\mu(0)=-1$ and $\mu(1)=1$.
Furthermore, if $t>1$, $|\mu(t)-1|<\frac{2\varepsilon}{\infnorm{b-a}}$,
and conversely, if $t<0$, $|\mu(t)+1|<\frac{2\varepsilon}{\infnorm{b-a}}$.
Consequently, $\phi(0)=a$, $\phi(1)=b$ and $\phi([0,1])$
is the line segment between $a$ and $b$, so $\phi([0,1])\subseteq X$.
Furthermore, if $t<0$, $\infnorm{a-\phi(t)}\leqslant
\left|\frac{1+\mu(t)}{2}\right|\infnorm{b-a}<\varepsilon$, and if $t>1$,
$\infnorm{b-\phi(t)}\leqslant\left|\frac{1-\mu(t)}{2}\right|\infnorm{b-a}<\varepsilon$.
We conclude from this analysis that $\phi(\R)\subseteq X_\varepsilon\subseteq U$.
It remains to show that $\phi\in\gval[\K]{}$. Using \lemref{lem:gpac_gen_op},
it suffices to show that $\tanh\in\gval[\K]{}$ and $\frac{1}{\tanh(\alpha)}\in\K$.
Since $\K$ is a field, we need to show that $\tanh(\alpha)\in\K$ which is a consequence
of $\K$ being a generable field and $\tanh$ being a generable function. We already
saw in Example \ref{ex:elem_func_gpac} that $\tanh\in\gval[\Q]{}\subseteq\gval[\K]{}$.

In the general case where $\gamma$ is a polygonal path, there are $0=t_1<t_2<\ldots<t_k=1$ such that
$\frestrict{\gamma}{[t_i,t_{i+1}]}$ is the line segment between $x_i=\gamma(t_i)$
and $x_{i+1}=\gamma(t_{i+1})$, furthermore we can always take $x_i\in\Q^n$.
Note that we can choose any parametrization for the path so in particular we can take $t_i=\frac{i}{k}$
and ensure that $t_i\in\Q$
for $i\in\intinterv{0}{k}$. Since by hypothesis $x_0,x_n\in\K^n$, we get that $x_i\in\K^n$ and $t_i\in\K$ for
all $i\in\intinterv{0}{k}$.

Let us denote by $\phi_{\varepsilon}^{a,b}$ the path
built in the previous case. We are simply going to add several instances of this path,
with the necessary shifting and scaling. Since the errors will sum up, we will increase
the approximation precision of each segment. Define
$\phi(t)=a+\sum_{i=1}^{k-1}\left(\phi_{\varepsilon/k}^{x_i,x_{i+1}}\left(\frac{t-t_i}{t_{i+1}-t_i}\right)-x_i\right)$
and consider the following cases:
\begin{itemize}
\item if $t<0$, then $\infnorm{\phi_{\varepsilon/k}^{x_i,x_{i+1}}\left(\frac{t-t_i}{t_{i+1}-t_i}\right)-x_i}<\frac{\varepsilon}{k}$
for all $i\in\intinterv{1}{k-1}$, so $\infnorm{a-\phi(t)}<\frac{k-1}{k}\varepsilon$ and $\phi(t)\in X_\varepsilon$
\item if $t\in[t_j,t_j+1]$ for some $j$, then
$\infnorm{\phi_{\varepsilon/k}^{x_i,x_{i+1}}\left(\frac{t-t_i}{t_{i+1}-t_i}\right)-x_i}<\frac{\varepsilon}{k}$
for all $i>j$, and conversely
$\infnorm{\phi_{\varepsilon/k}^{x_i,x_{i+1}}\left(\frac{t-t_i}{t_{i+1}-t_i}\right)-x_{i+1}}<\frac{\varepsilon}{k}$
for all $i<j$. Finally $u=\phi_{\varepsilon/k}^{x_j,x_{j+1}}\left(\frac{t-t_j}{t_{j+1}-t_j}\right)$ belongs
to the line segment from $x_j$ to $x_j+1$. Since $a=x_1$, we get that $\infnorm{u-\phi(t)}\leqslant\frac{k-1}{k}\varepsilon$
and thus $\phi(t)\in X_\varepsilon$.
\item if $t>1$ then $\infnorm{b-\phi(t)}<\varepsilon$ for the same reason as $t<0$, and
thus $\phi(t)\in X_\varepsilon$.
\end{itemize}
We conclude that $\phi(\R)\subseteq X_\varepsilon\subseteq U$ and one easily checks
that $\phi(0)=a$ and $\phi(1)=b$. Furthermore $\phi\in\gval[\K]{}$ by \lemref{lem:gpac_gen_op}
and because the $x_i$ and $t_i$ belong to $\K$ (see the details in the case of the trivial path).
\end{proof}

The immediate corollary of this result is that $\K$ is also preserved by multidimensional
generable functions. Indeed, by composing a multidimensional function with a
unidimensional one, we get back to the unidimensional case and conclude that
any generable point in the input domain must have a generable image.

\begin{corollary}[Generable field stability]\label{cor:generable_field_ext_stab}
Let $(f:\subseteq\R^d\rightarrow\R^\ell)\in\gval[\K]{}$, then $f(\K^d\cap\dom{f})\subseteq\K^\ell$.
\end{corollary}

\begin{proof}
Apply Definition \ref{def:gpac_generable_ext} to get $n\in\N$, $p\in\MAT{n}{d}{\K}[\R^n]$,
$x_0\in\dom{f}\cap\K^d$, $y_0\in\K^n$ and $y:\dom{f}\rightarrow\R^n$. Let $u\in\dom{f}\cap\K^d$.
Since $\dom{f}$ is open and connected, by \propref{prop:connected_is_generable_connected},
there exists $(\gamma:\R\rightarrow\dom{f})\in\gval{}$ such that $\gamma(0)=x_0$
and $\gamma(1)=u$. Apply Definition \ref{def:gpac_generable_ext} to $\gamma$ to get
$\bar{n}\in\N$, $\bar{p}\in\MAT{\bar{n}}{1}{\K}[\R^{\bar{n}}]$, $\bar{x}_0\in\K$, $\bar{y}_0\in\K^{\bar{n}}$
and $\bar{y}:\R\rightarrow\R^{\bar{n}}$. Define $z(t)=y(\gamma(t))=y(\bar{y}_{1..d}(t))$,
then $z'(t)=\jacobian{y}(\gamma(t))\gamma'(t)=p(y(\gamma(t)))\gamma'(t)=p(z(t))\bar{p}_{1..d}(\bar{y}(t))$
and $z(0)=y(\gamma(0))=y(x_0)=y_0$. In other words $(\bar{y},z)$ satisfy:
\[\left\{\begin{array}{@{}r@{}l}\bar{y}(0)&=x_0\in\K^d\\z(0)&=y_0\in\K^n\end{array}\right.
\qquad\left\{\begin{array}{@{}r@{}l}\bar{y}'&=\bar{p}(\bar{y})\\z'&=p(z)\bar{p}_{1..\ell}(\bar{y})\end{array}\right.\]
Consequently $(z:\R\rightarrow\R^\ell)\in\gval{}$ so, by definition of a generable field,
$z(\K)\subseteq\K^zell$.
Conclude by noticing that $z(1)=y(\gamma(1))=y(u)$.
\end{proof}

\subsection{Generable real numbers}

In this section, we formalize the notion of generable field with an operator
and study its properties. Recall that the smallest field we are looking for is a
subset of $\R$ but it must also contains $\Q$. We consider the following
operator $G$ on subset of real numbers.

\[G:\left\{\begin{array}{ccc}
\powerset{\R}&\rightarrow&\powerset{\R}\\
X&\mapsto&\displaystyle\bigcup_{f\in\gval[X]{}}f(X)
\end{array}\right.\]

\begin{remark}[$G$ monotone and non-decreasing]
One can check that $G$ is monotone ($X\subseteq G(X)$
for any $X\subseteq\R$). Indeed for any $x\in X$, the constant function
$\lambdafunex{u}{x}$ belongs to $\gval[X]{}$. Moreover, it is non-decreasing
because $\gval[X]{}\subseteq\gval[Y]{}$ if $X\subseteq Y$.
\end{remark}

It is clear that by definition, a field is generable if and only if it is $G$-stable.
An interesting property of $G$ is that its definition can be simplified.
More precisely, by rescaling the functions, we can always assume that the image
of $G$ is produced by the evaluation of generable functions at a particular
point, say $1$, instead of the entire field.

\begin{lemma}[Alternative definition of $G$]\label{lem:alt_def_G}
If $X$ is a field then,
\[G(X)=\big\{f(1):f\in\gval[X]{}\big\}\]
\end{lemma}

\begin{proof}
Let $x\in G(X)$, then there exists $f\in\gval[X]{}$
and $t\in X$ such that $x=f(t)$. Consequently there exists $d\in\N$, $y_0\in X^d$,
$p\in X^d[\R^d]$ and $y:\R\rightarrow\R^d$ satisfying
Definition \ref{def:gpac_generable}:
\begin{itemize}
\item $y'=p(y)$ and $y(0)=y_0$
\item $y_1=f$
\end{itemize}
Consider $g(u)=f(ut)$ and note that $g(1)=f(t)=x$.
We will see that $g\in\gval[X]{}$.
Indeed, consider $z(u)=y(tu)$ then for all $u\in\R$:
\begin{itemize}
\item $z(0)=y(0)=y_0\in X^d$;
\item $z'(u)=ty'(tu)=tp(z(u))=q(z(u))$ where $q=tp$ is a polynomial with coefficients
in $X$ since $t\in X$ and $X$ is a field
\item $z_1(u)=y_1(tu)=g(u)$
\end{itemize}
\end{proof}

A consequence of this alternative definition is a simple proof that $G$ preserves
the property of being a field. This will turn out to be crucial fact later on.

\begin{lemma}[$G$ maps fields to fields]\label{lem:G_preserves_field}
If $X$ is a field, then $G(X)$ is a field.
\end{lemma}

\begin{proof}
Let $x,y\in G(X)$, by \lemref{lem:alt_def_G} there exists
$f,g\in\gval[X]{}$ such that $x=f(1)$ and $y=g(1)$. Apply \lemref{lem:gpac_gen_op}
to get that $f\pm g$ and $fg$ belong to $\gval[X]{}$
And thus $x\pm y$ and $xy$ belong to $G(X)$.

Finally the case of $\frac{1}{x}$ (when $x\neq0$) is slightly more subtle: we cannot
simply compute $\frac{1}{f}$ because $f$ may cancel. Instead we are going to compute
$\frac{1}{g}$ where $g(1)=f(1)$ but $g$ nevers cancels.

First, note that we can always assume that $x>0$ because $G(X)$ is closed under the negation,
and $-\frac{1}{x}=\frac{1}{-x}$. Since $f(1)=x>0$ and $f$ is continuous,
it means there exists $\varepsilon>0$ such that $f(t)>0$ for all $t\in[1-\varepsilon,1+\varepsilon]$
and we can take $\varepsilon\in\Q$. Define $g(t)=f(t)+\big(1+f(t)^2\big)\left(\frac{t-1}{\varepsilon}\right)^2$.
It is not hard to see that $g(1)=f(1)$ and that $g(t)>0$ for all $t\in\R$.
Furthermore, $g\in\gval[X]{}$ because of \lemref{lem:gpac_gen_op}. Note
that we use the part of the lemma which does not assume that $X$ is a generable field!

Using \lemref{lem:gpac_gen_op}, we conclude that $\frac{1}{g}\in\gval[X]{}$
and thus $\frac{1}{x}\in G(X)$.
\end{proof}

Not only $G$ maps fields to fields, but it also preserves polynomial-time computability.
This is of major interest to us to show that there exists a generable field with
low complexity numbers. Here $\Rpoly$ denotes the set of polynomial
time computable real numbers \cite{Ko91}.

\begin{lemma}[$G$ preserves polytime computability]\label{lem:G_preserves_Rpoly}
$G$ maps subsets of polynomial time computable real numbers into themselves, \emph{i.e.}
for any $X\subseteq\Rpoly$, $G(X)\subseteq\Rpoly$.
\end{lemma}

\begin{proof}
Let $X\subseteq\Rpoly$ and $x\in G(X)$,
$f\in\gval[X]{}$ and $t\in X$ such that $x=f(t)$. We can use
\cite{BGP12}
to conclude that $x$ is polynomial time computable, thus $x\in\Rpoly$.
\end{proof}

Finally, the core of what makes $G$ very special is its finiteness property.
Essentially, it means that if $x\in G(X)$ then $x$ really only requires a finite
number of elements in $X$ to be computed. In the framework of order and lattice theory,
this shows that $G$ is a Scott-continuous function between the complete partial
order (CPO) $(\mathcal{L},\subseteq)$ and itself.

\begin{lemma}[Finiteness of $G$]\label{lem:finiteness_G}
For any $X\subseteq\R$ and $x\in G(X)$, there exists a finite $Y\subseteq X$ such
that $x\in G(Y)$.
\end{lemma}

\begin{proof}
Let $x\in G(X)$, then there exists $f\in\gval[X]{}$ and $t\in X$
such that $x=f(t)$. Then there exists $y_0\in X^d$ and a polynomial $p$ with
coefficients in $X$ such that $f$ satisfies Definition \ref{def:gpac_generable}.
Define $Y$ as the subset of $X$ containing $t$, the components of $y_0$ and all
the coefficients of $p$. Then $Y$ is finite and $f\in\gval[Y]{}$.
Furthermore $t\in Y$ so $x\in G(Y)$.
\end{proof}

We can now define the set of ``generable real numbers'', call it $\Rgen$.
The main result of this section is that $\Rgen$ is the smallest generable field.
But more surprisingly, we show that all the elements of $\Rgen$ are
polynomial time computable (in the sense of Computable Analysis).

\begin{definition}[Generable real numbers]\label{def:Rgen}
\[\Rgen=\bigcup_{n\geqslant0}\fiter{G}{n}(\Q).\]
\end{definition}

\begin{theorem}[$\Rgen$ is generable subfield of $\Rpoly$]\label{th:Rgen_in_Rpoly}
$\Rgen$ is the smallest generable field for inclusion. Furthermore, it form a generable
subfield of polynomial time computable real numbers in the sense of Computable Analysis,
i.e. $\Rgen\subseteq\Rpoly$.
\end{theorem}

\begin{proof}
First observe that any generable field must contain $\Rgen$. Indeed, let $\K$
be a generable field: then $G(\K)\subseteq\K$ by definition. But $G$ is non-decreasing
thus $G(\Q)\subseteq G(\K)\subseteq\K$. By applying $G$ repeatedly, we get that
$\fiter{G}{n}(\Q)\subseteq\K$ for all $n$. Thus $\Rgen\subseteq\K$.

Conversely, we need to show that $\Rgen$ is a field. Observe that since $G$ is
monotone, $\fiter{G}{n}(\Q)$ is an increasing sequence (for inclusion). Let $x,y\in\Rgen$,
then there exists $n\in\N$ such that $x,y\in\fiter{G}{n}(\Q)$. Apply \lemref{lem:G_preserves_field}
to get that $\fiter{G}{n}(\Q)$ is a field. It follows that $x+y,x-y,xy$ and $\frac{x}{y}$
(if $y\neq0$) belong to $\fiter{G}{n}(\Q)\subseteq\Rgen$. Thus $\Rgen$ is a field.

It remains to show that $\Rgen$ is a \emph{generable} field. This follows
from \lemref{lem:finiteness_G}: let $x\in G(\Rgen)$, then there exists a \textbf{finite}
$Y\subseteq\Rgen$ such that $x\in G(Y)$. Using the same reasoning as above,
there exists $n\in\N$ such that $Y\subseteq\fiter{G}{n}(\Q)$. Thus
$x\in G(Y)\subseteq G(\fiter{G}{n}(\Q))=\fiter{G}{n+1}(\Q)\subseteq\Rgen$. It
follows that $G(\Rgen)\subseteq\Rgen$, \emph{i.e.} it is generable.

Finally, since $\Q\subseteq\Rpoly$, iterating \lemref{lem:G_preserves_Rpoly} yields that
$\fiter{G}{n}(\Q)\subseteq\Rpoly$ for all $n\in\N$ and thus $\Rgen\subseteq\Rpoly$.
\end{proof}

\section*{References}

\bibliographystyle{elsarticle-harv}
\bibliography{
  ContComp
}

\end{document}